\documentclass[a4paper,11pt]{article}

\usepackage{graphicx}
\usepackage{amsmath,amsfonts}
\usepackage{amssymb}
\usepackage{color}
\usepackage{algorithm}
\usepackage{algpseudocode}
\usepackage{booktabs}
\usepackage{url}
\usepackage[subrefformat=parens]{subcaption}
\newenvironment{proof}[1][Proof]{\textbf{#1.} }{\hfill$\square$}
\newtheorem{theorem}{Theorem}
\newtheorem{lemma}{Lemma}
\newtheorem{remark}{Remark}
\newtheorem{ass}{Assumption}

\newtheorem{prob}{Problem}

\newcommand{\A}{{\mathcal A}}
\newcommand{\B}{{\mathcal B}}

\newcommand{\C}{{\mathcal C}}
\newcommand{\D}{{\mathcal D}}

\newcommand{\Hm}{{\mathbf H}}

\begin{document}

\title{Cyclic Reformulation-Based Identification and Polytopic Uncertainty Modeling for Multirate Systems\thanks{Corresponding author: H.~Okajima. Email: okajima@mech.kumamoto-u.ac.jp}}

\author{Hiroshi Okajima, Kakeru Ono\\
{\small Faculty of Advanced Science and Technology, Kumamoto University, Japan}}

\date{}

\maketitle


\begin{abstract}
Modern control systems increasingly rely on heterogeneous sensors operating at different sampling rates, where intermittently missing outputs pose fundamental challenges for system identification. This paper proposes a non-iterative, control-oriented identification method for multirate systems based on cyclic reformulation. The method transforms multirate data into an expanded time-invariant representation and yields $M$ parameter sets from a single input--output dataset, where $M$ is the least common multiple of the sensor sampling periods. These parameter sets, expressed in a common state coordinate of the original dimension, are used in two complementary ways: their centroid, in which each output row is averaged over the phases at which that output is observed, serves as a noise-reduced nominal model, while their convex hull gives a polytopic uncertainty model whose vertices can be used directly in vertex-based LMI robust control design. Building on the exact cyclic-structure recovery established in the authors' preceding work, the paper contributes the centroid and polytopic models and evaluates them under finite noise through output-level validation statistics. Numerical simulations support both models: an illustrative SISO example shows that the centroid attains higher validation FIT than the best individual vertex and substantially outperforms an interpolation-based baseline, while a MIMO multirate sensing example confirms that the constructed polytope contains models whose validation FIT exceeds 95\% on average even at the highest tested noise level. The resulting polytope is a data-derived uncertainty description obtained from a single experiment, so that the framework links multirate system identification with robust-control-oriented uncertainty modeling without iterative EM-type optimization.
\end{abstract}

\section{Introduction}\label{sec1}

System identification plays a crucial role in model-based control design, where accurate mathematical models directly determine control performance~\cite{id0,id1,id2,katayama}. Recent advances have extended identification techniques to nonlinear systems~\cite{non1}, periodic systems~\cite{lpv,lptv}, closed-loop systems~\cite{maruta2022}, and various practical applications.

In modern control systems, feedback control using multiple sensors is common. When these sensors are based on different physical principles, they often operate at different sampling rates, as each sensing principle has its own achievable rate under hardware and communication constraints. Such multirate sensing environments arise frequently in mobile robot control~\cite{mobile1,mobile2} and sampled-data control systems~\cite{chen1995sampled}. Since multirate systems can be viewed as systems with periodically unavailable output measurements, their identification is significantly more challenging than that of standard linear time-invariant systems.

Several approaches have been proposed for multirate system identification. Subspace-based identification for non-uniformly sampled multirate data was developed in~\cite{li2006subspace}, and Kalman filter design for fault detection in such systems was addressed in~\cite{li2008kalman}. Lifting techniques~\cite{rev1-1,li2001fast} enable fast-rate model identification but face difficulties in recovering original system parameters due to variable products in the expanded state space. Modified subspace methods for periodically non-uniformly sampled systems were also proposed in~\cite{ding2014modified}. Stochastic gradient approaches offer another direction: partially coupled algorithms for non-uniformly sampled systems were developed in~\cite{ding2010partially}, avoiding the need to estimate the full parameter vector simultaneously. Like the lifting techniques~\cite{rev1-1,li2001fast}, a recent lifting-based subspace method for coprime non-uniform sensing~\cite{hu2024beyond} recovers the single-rate model only through a subsequent eigenvalue decomposition of the lifted model. A further route recovers the missing outputs before identification, by interpolation or by low-rank (nuclear-norm) completion~\cite{hou2024missing}, and then applies a standard single-rate method; such fill-then-identify approaches deliver a single nominal model.

An alternative approach treats missing outputs as latent variables and applies the Expectation-Maximization (EM) algorithm~\cite{dempster1977}, which has been applied to multirate processes with random delays~\cite{xie2013fir}, missing output data in linear parameter-varying systems~\cite{xiong2014multiple}, and state-space models with time delays~\cite{chen2019state}. Separately, robust maximum-likelihood methods~\cite{gibson2005robust} provide improved convergence properties for multivariable systems through direct likelihood optimization. While both classes of iterative methods can yield statistically efficient estimates under appropriate assumptions, they generally require careful initialization and may encounter local convergence issues in practice.

The approach taken in this paper originates in the authors' preceding work, which is based on cyclic reformulation~\cite{cyc1,cyc2}. Because the sensing pattern repeats with period $M$, the least common multiple of the sensor sampling periods, a multirate system can be regarded as a periodically time-varying system. Lifting and cyclic reformulation both convert such a period-$M$ system into a time-invariant representation; in the latter, the original system matrices remain sub-blocks of the expanded matrices. Building on this property, the authors developed a cyclic reformulation-based identification method for linear periodically time-varying (LPTV) systems~\cite{okajima1}, which recovers the true system parameters exactly (up to coordinate transformation) under noise-free conditions, and then extended it to multirate sensing environments~\cite{okajima2}, where the sensing timing of each output channel is characterized by a selection matrix and the resulting period-$M$ system is identified via cyclic reformulation. In contrast to the lifting- and EM-based methods discussed above, this formulation is exact and non-iterative and directly recovers the original system matrices. The focus of~\cite{okajima2}, however, was noise-free identification, where the $M$ parameter sets obtained from cyclic reformulation were shown to coincide exactly and hence to yield a single model. Under noisy conditions the $M$ parameter sets exhibit different perturbations, and the systematic treatment of this discrepancy was not addressed in~\cite{okajima2}. For LTI systems, the authors have recently shown in~\cite{okajima3} that such noise-induced dispersion can be turned into useful information: an LTI plant is identified through cyclic reformulation with a period $N$ chosen by the user, and the $N$ resulting parameter sets, which disperse under noise, are taken as the vertices of a polytopic uncertainty model. Whether and how this construction carries over to multirate systems, where some output rows are structurally absent at certain phases, has not been examined.

This paper addresses this gap by exploiting the noise-induced discrepancy among the $M$ parameter sets as a source of uncertainty information. Cyclic reformulation itself~\cite{cyc1,cyc2}, its use for identification~\cite{okajima1,okajima2}, and the LTI polytope construction of~\cite{okajima3} are inherited; the new elements of the present paper are (a) the use of the phase-wise dispersion of the $M$ parameter sets under finite noise as uncertainty information for multirate systems, (b) a row-wise centroid model that averages each output row only over the phases at which that output is observed, and (c) a polytopic uncertainty model with $M$ vertices expressed in a common state coordinate, which, unlike the LTI construction of~\cite{okajima3}, requires completing the output rows that are structurally empty at non-observation phases with the row-wise centroid so that all $M$ parameter sets become well-defined vertices. To the best of the authors' knowledge, no existing method simultaneously provides a nominal model and a polytopic uncertainty description from a single multirate experiment without iterative computation, and the resulting polytope is directly compatible with vertex-based LMI robust control design~\cite{robust1,robust3,robust4,kothare1996robust}. In this sense the present paper supplies the multirate counterpart of the polytopic modeling pipeline of~\cite{okajima3}, and its finite-noise modeling step is kept distinct from the exact noise-free recovery of~\cite{okajima1,okajima2}.

This paper is organized as follows. Section~\ref{sec0} presents mathematical preliminaries on cyclic reformulation and polytopic uncertainty models. Section~\ref{sec2} formulates multirate systems as periodically time-varying systems. Section~\ref{sec3} develops the cyclic reformulation for multirate systems. Section~\ref{sec4} proposes the system identification algorithm, including the centroid and polytopic uncertainty models. Section~\ref{sec5} demonstrates the effectiveness through numerical simulations. Section~\ref{sec_disc} positions the proposed method against existing approaches and discusses the results and limitations. Section~\ref{sec6} concludes the paper.

\section{Mathematical Preliminaries}\label{sec0}

\subsection{Cyclic Reformulation of Periodically Time-varying Systems} \label{sec01}
This section explains cyclic reformulation, a time-invariant system expression method for periodically time-varying systems, as a preliminary step~\cite{cyc1,cyc2}. Consider a discrete-time linear periodically time-varying system with period $M$ (where $M$ is a positive integer) represented by $(A_k,B_k,C_k,D_k)$:
\begin{eqnarray} 
x(k+1)&=&A_k x(k)+B_k u(k)\label{siki1-jizen}\\
y(k)&=&C_k x(k)+ D_k u(k)\label{siki12-jizen}
\end{eqnarray}

For a positive integer $q$, define the cyclic shift matrix $\check{S}_q \in \mathbb{R}^{Mq \times Mq}$ as
\begin{equation}
     \check{S}_q = \left[\begin{array}{ccccc}
           \mathbf{0}_{q \times q} & I_q &  \mathbf{0}_{q \times q} & \cdots &  \mathbf{0}_{q \times q}\\
           \mathbf{0}_{q \times q} &  \mathbf{0}_{q \times q} & I_q & \ddots & \vdots\\
          \vdots & \ddots & \ddots & \ddots &  \mathbf{0}_{q \times q}\\
           \mathbf{0}_{q \times q} & \ddots & \ddots & \ddots & I_q\\
          I_q &  \mathbf{0}_{q \times q} & \cdots & \cdots &  \mathbf{0}_{q \times q} 
     \end{array} \right].   \label{sq}
\end{equation}
The matrix $\check{S}_q$ is nonsingular and satisfies $\check{S}_q^M = I_{Mq}$, i.e., $\check{S}_q$ is an $M$-th root of the identity matrix; this property is used in subsequent derivations involving $\check{S}_q^{Mi}$ for integer $i$. Its inverse $\check{S}_q^{-1}$ shifts block components in the opposite direction. For any block diagonal matrix $E \in \mathbb{R}^{Mq\times Mq}$ with $q\times q$ block elements $E_i$, $\check{S}_q^{-1}E \check{S}_q$ is also block diagonal, with the individual blocks shifted by one position relative to $E$.

Let $e_r \in \mathbb{R}^M$ ($r = 0, \ldots, M-1$) denote the $(r+1)$-th standard basis vector of $\mathbb{R}^M$, and let $\otimes$ denote the Kronecker product. The cycled input signal $\check{u}(k) \in \mathbb{R}^{Mm}$ is then constructed by
\begin{eqnarray}
\check{u}(k) = e_{k \bmod M} \otimes u(k), \label{checku}
\end{eqnarray}
so that $\check{u}(k)$ has the unique non-zero sub-vector $u(k)$ at the block position $k \bmod M$, while the rest are zeros, and the non-zero block cyclically shifts along the column blocks.

The cyclic reformulation of the $M$-periodic system (\ref{siki1-jizen}), (\ref{siki12-jizen}) is given by
\begin{equation}
 \label{pre0-eqcyclic}
     \begin{array}{rcl}
          \check{x}(k+1) & = & \check{A}\check{x}(k) + \check{B}\check{u}(k)\\
          \check{y}(k) & = & \check{C}\check{x}(k) + \check{D}\check{u}(k),
     \end{array}
\end{equation}
where
\begin{eqnarray}
\check{A} &=& \check{S}_n^{-1} \mathrm{diag}(A_0, \ldots, A_{M-1}), \nonumber \\
\check{B} &=& \check{S}_n^{-1} \mathrm{diag}(B_0, \ldots, B_{M-1}), \nonumber \\
\check{C} &=& \mathrm{diag}(C_0, \ldots, C_{M-1}), \nonumber \\
\check{D} &=& \mathrm{diag}(D_0, \ldots, D_{M-1}). \label{pre0-checkABCD}
\end{eqnarray}
The dimensions are $\check{A}\in \mathbb{R}^{Mn\times Mn}$, $\check{B}\in \mathbb{R}^{Mn\times Mm}$, $\check{C}\in \mathbb{R}^{Ml\times Mn}$, $\check{D}\in \mathbb{R}^{Ml\times Mm}$, $\check{x}(k)\in \mathbb{R}^{Mn}$, and $\check{y}(k)\in \mathbb{R}^{Ml}$. Premultiplication by $\check{S}_n^{-1}$ shifts the block diagonal entries into the subdiagonal positions characteristic of cyclic reformulation, so that $\check{A}$ and $\check{B}$ are cyclic matrices while $\check{C}$ and $\check{D}$ retain block diagonal form. Although the cyclic reformulated system (\ref{pre0-eqcyclic}) is itself linear time-invariant, direct step-by-step computation using (\ref{pre0-checkABCD}) and the cycled input (\ref{checku}) shows that, for every $k$, the cycled state $\check{x}(k)$ and cycled output $\check{y}(k)$ each have a unique non-zero sub-vector at block position $k \bmod M$ that coincides exactly with the state $x(k)$ and output $y(k)$ of the original periodically time-varying system (\ref{siki1-jizen}), (\ref{siki12-jizen}).

Furthermore, in~\cite{okajima1}, the properties of Markov parameters for the cyclic reformulated system were characterized as follows. Markov parameters $\check{H}(i)$ are coefficients of the impulse response of system (\ref{pre0-eqcyclic}) and are given by
\begin{eqnarray}
\check H(i) = \begin{cases}\check D,&i = 0 \\ \check C \check A^{i-1} \check B, & i=1,2,\cdots \end{cases}\label{markov}
\end{eqnarray}
Using the cyclic shift matrix $\check{S}_q$, the following important lemma for Markov parameters $\check H(i)$ of the cycled system (\ref{pre0-eqcyclic}) is obtained~\cite{okajima1}:
\begin{lemma}\label{lemma3}
Consider the following $Ml\times Mm$ matrix:
\begin{eqnarray}
 \check S_l^i \check H(i+j)\check S_m^{j}
\end{eqnarray}
Then, $\check S_l^i \check H(i+j)\check S_m^j$ is a block diagonal matrix with $l\times m$ block elements, where all off-diagonal blocks are zero matrices, for any non-negative integers $i, j$. In addition, for any non-negative integer $i$ and any positive integer $j$, the following matrix:
\begin{eqnarray}
 \check S_l^i \check H(i+j)\check S_m^{j-1}
\end{eqnarray}
is a cyclic matrix. \hfill $\Box$
\end{lemma}
The characteristics of cyclic reformulation shown in Lemma \ref{lemma3} are the basis of the identification method of~\cite{okajima1} and of the present paper.

\subsection{Polytopic Uncertainty Model}\label{sec02}
We describe a system representation with polytopic uncertainty using $N$ parameter sets as vertices:
\begin{eqnarray} 
x(k+1)&=&A(\lambda) x(k)+B(\lambda) u(k)\label{siki1-jizen2}\\
y(k)&=&C(\lambda) x(k)+ D(\lambda) u(k)\label{siki2-jizen2}
\end{eqnarray}
where $A(\lambda), B(\lambda), C(\lambda), D(\lambda)$ are given in the following form:
\begin{eqnarray}
\label{eq:polytope_construction}
A(\lambda) &=& \sum_{i=0}^{N-1} \lambda_i A_{mi}, 
B(\lambda) = \sum_{i=0}^{N-1} \lambda_i B_{mi}, \\
C(\lambda) &=& \sum_{i=0}^{N-1} \lambda_i C_{mi}, 
D(\lambda) = \sum_{i=0}^{N-1} \lambda_i D_{mi} \label{eq:polytope_construction_cd}
\end{eqnarray}
Here, $\lambda \in \Lambda := \{\lambda \in \mathbb{R}^N : \lambda_i \geq 0, \sum_{i=0}^{N-1} \lambda_i = 1\}$. The polytopic uncertainty model is given as a convex polytope in the system parameter space, and any point inside the polytope is represented by the convex combinations in (\ref{eq:polytope_construction}), (\ref{eq:polytope_construction_cd}). The polytopic uncertainty model expressed by (\ref{siki1-jizen2}), (\ref{siki2-jizen2}) is widely used in robust control~\cite{robust1,robust3,robust4,kothare1996robust}. In particular, it has good compatibility with linear matrix inequalities (LMIs). Specifically, if LMI conditions are satisfied for each vertex of the polytope under the setting of a common Lyapunov matrix---provided that the LMI conditions are convex in the polytope parameter $\lambda$ and vertex-independent decision variables are used---the same conditions are guaranteed for all points inside due to convexity. Due to these characteristics, it can be applied to various robust control methods including $H_\infty$ control, robust pole placement, and robust model predictive control.


\section{Problem Formulation}\label{sec2}
This section describes the system representation of multirate systems. As a preliminary step before addressing multirate systems, we first consider the plant as a discrete-time linear time-invariant system. The linear time-invariant system is expressed as follows:
\begin{eqnarray} 
x(k+1)&=&Ax(k)+Bu(k)+d_u(k)\label{siki1}\\
y(k)&=&Cx(k)+ Du(k)+d_y(k)\label{siki12}
\end{eqnarray}
where $k$ is the discrete time, $x(k) \in \mathbb{R}^{n}$ is the state,
$u(k) \in \mathbb{R}^{m}$ is the input, $y(k) \in \mathbb{R}^{l}$ is the output,
and $d_u(k) \in \mathbb{R}^{n}$, $d_y(k) \in \mathbb{R}^{l}$ are process noise 
and observation noise, respectively.
The system matrices $A \in \mathbb{R}^{n\times n}$, $B \in \mathbb{R}^{n\times m}$, 
$C \in \mathbb{R}^{l\times n}$, and $D \in \mathbb{R}^{l\times m}$ are unknown 
and to be identified from input--output data.

The following assumptions are made throughout this paper:
the pair $(A,B)$ is controllable, $A$ has full rank $n$,
the system order $n$ is known,
the input $u(k)$ is persistently exciting of a sufficiently high
order as required by the subspace identification method,
and the noise signals $d_u(k)$ and $d_y(k)$
are zero-mean and mutually uncorrelated.
In practice, $n$ is either known from physical modeling or estimated
from the dominant singular values in the subspace identification step
and treated as known thereafter, which fixes the expanded order $Mn$
of the cycled realization.
An observability-type condition involving the multirate output
selection $V_j$ is stated separately as Assumption~\ref{ass2} below.
No specific distribution is assumed in the theoretical development,
while Gaussian distributions are used for numerical verification
in Section~\ref{sec5}.

Next, based on the linear time-invariant system (\ref{siki1}), (\ref{siki12}), we consider the representation of a multirate system where output observation periods differ from the control period. Using the periodically time-varying matrix $V_k$, the multirate system is represented by the following state-space equations:
\begin{eqnarray}
x(k+1)&=&Ax(k)+Bu(k)+d_u(k)\label{eq:P_multi}\\
y(k)&=&V_kCx(k)+V_kDu(k) + V_k d_y(k) \label{eq:P_multi2}
\end{eqnarray}
The observation periods of the outputs $y_i(k)$, which are the components in $y(k)$, can be different. Each output observation period is set as a natural number multiple of the control period, and $M_1, \cdots, M_l$ are the observation periods of outputs $y_1(k), \cdots, y_l(k)$. Taking output $y_1(k)$ as an example, $y_1(k)$ is observed once every $M_1$ steps. At non-observation times, the measured signal used for identification is set to zero, i.e., the corresponding zero-padded measurement is $0$; the underlying physical output itself is not assumed to be zero. This zero-padded output sequence is then used for system identification.

The periodically time-varying matrix $V_k$ is determined as follows to characterize the above-mentioned observation period~\cite{multi1,multi2}. When $M$ is the least common multiple of the observation periods $M_1, \cdots, M_l$, $V_k$ is a matrix that characterizes the observation period of the output, and its period is given as $M$:
\begin{eqnarray}\label{eq:Sk}
V_{k}=\mathrm{diag}\left[\begin{array}{ccccc}
 v_{1}(k) & \cdots & v_{i}(k) & \cdots & v_{l}(k) 
\end{array}\right]
\end{eqnarray}
The element $v_i(k)$ in the matrix corresponds to $y_i(k)$ with the $i$-th observation period $M_{i}$. The component $v_i(k)$ is $1$ when $y_i$ is observed at step $k$ and $0$ otherwise. Since every $M_i$ divides $M$, $V_k$ is $M$-periodic and $v_i(k)=1$ exactly $M/M_i$ times per period. From the above, the observation timing of all outputs can be represented within a framework with period $M$. Therefore, the following equation holds, and a multirate system can be represented by providing $\{V_0, \cdots, V_{M-1}\}$:
\begin{eqnarray}\label{eq:Sk2}
V_{k}=V_{k\,\bmod\,M}
\end{eqnarray}

Fig.~\ref{fig:multirate_sampling} illustrates the multirate sampling pattern for the case $M_1 = 2$, $M_2 = 3$ ($M = 6$). The input $u(k)$ is applied at every control step, while the two sensor outputs $y_1(k)$ and $y_2(k)$ are observed at their respective periods. At non-observation timings, the corresponding zero-padded measurement used for identification is set to zero. The bracket below the time axis indicates one full period $M = 6$, after which the observation pattern repeats.
\begin{figure}[!t]
\centering
\includegraphics[width=0.75\textwidth]{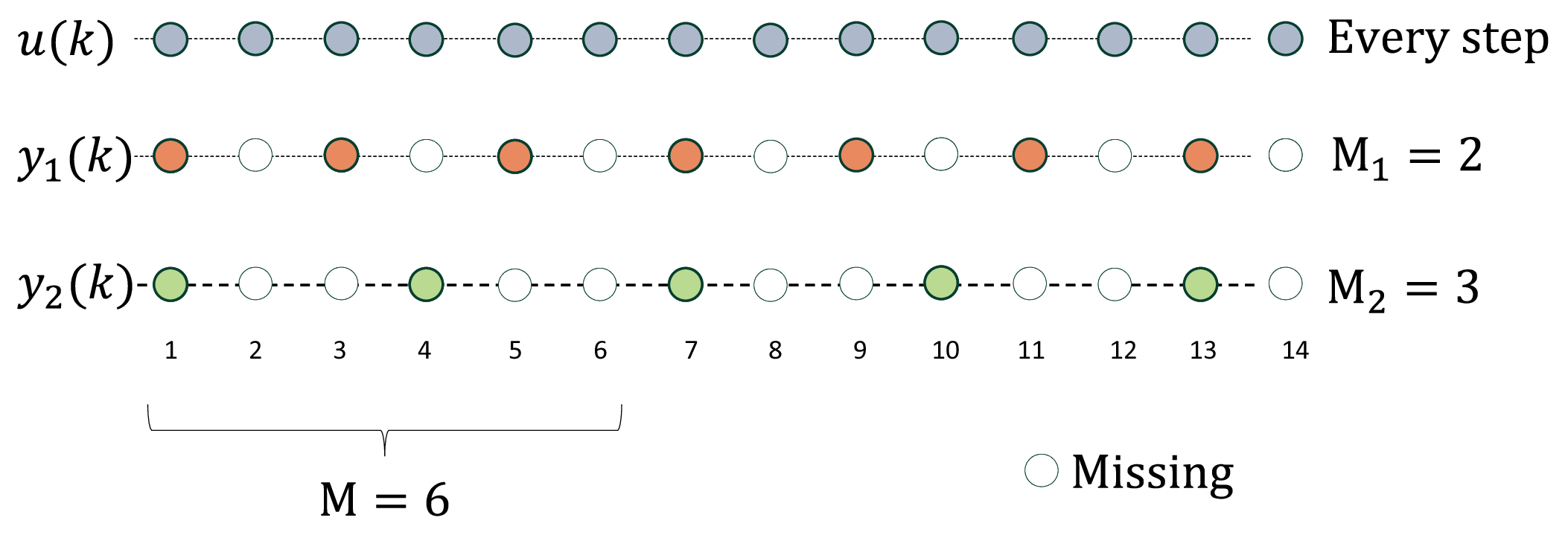}
\caption{Multirate sampling pattern for $M_1 = 2$, $M_2 = 3$ ($M = 6$). Filled circles indicate observed samples; open circles indicate missing outputs set to zero.}
\label{fig:multirate_sampling}
\end{figure}

Next, we provide the following assumption regarding the observability of the multirate system expressed as an $M$-periodic system:

\begin{ass}\label{ass2}
There exists at least one $j \in \{0, 1, \ldots, M-1\}$ such that the pair $(V_j C, A^M)$ is observable, i.e., the observability matrix for $(V_j C, A^M)$ has rank $n$.
\end{ass}

\begin{remark}\label{rem:ass2}
Assumption~\ref{ass2} is imposed directly on the pair $(V_j C, A^M)$ rather than on $(V_j C, A)$, because observability of $(V_j C, A)$ together with full-rank $A$ implies observability of $(V_j C, A^M)$ only under an additional spectral non-aliasing condition, namely that distinct eigenvalues of $A$ remain distinct after the $M$-th power~\cite{chen1995sampled,cyc2}. In typical multirate sensor fusion applications involving complementary sensors (e.g., vision and IMU, or position and velocity measurements), at least one observation timing provides sufficient sensor outputs and the spectrum of $A$ avoids exact aliasing, so that Assumption~\ref{ass2} is naturally satisfied.
\end{remark}

From the above, the multirate system (\ref{eq:P_multi}), (\ref{eq:P_multi2}) is represented as a periodically time-varying system with period $M$. By adopting this representation, the characteristics of the multirate system are incorporated into $V_k$, and the parameters to be identified in mathematical modeling are reduced to four: $A, B, C, D$.

\begin{prob}\label{prob1}
Suppose that input data $\{u(k)\}$ is applied to the plant (\ref{eq:P_multi}), (\ref{eq:P_multi2}) under a multirate sensing environment, yielding output data $\{y(k)\}$, and let $N_{\text{data}}$ be the data length. From the input--output data $\{u(k), y(k)\}_{k=0}^{N_{\text{data}}-1}$, estimate the system matrices $A$, $B$, $C$, and $D$.
\hfill $\Box$
\end{prob}

We denote the system matrices $A_m$, $B_m$, $C_m$, and $D_m$ as the solutions of Problem \ref{prob1}. Among these, $A_m$, $B_m$, and $C_m$ are identifiable only up to a common state coordinate transformation, whereas $D_m$ is invariant under such a transformation and is therefore coordinate-independent (and, in the noise-free case, recovered exactly). The procedure of Sections~\ref{sec3} and~\ref{sec4} yields $M$ parameter sets $(A_{mi}, B_{mi}, C_{mi}, D_{mi})$, $i=0,\ldots,M-1$, one for each phase of the cyclic reformulation~\cite{okajima2}; following the LTI construction of~\cite{okajima3}, each of them is regarded as an estimate of $(A, B, C, D)$ in this sense.

\section{Cyclic Reformulation of Multirate Systems}\label{sec3}

\subsection{Time-Invariant Representation via Cyclic Reformulation} \label{sec22}

This section applies the cyclic reformulation of Section~\ref{sec01} to the multirate system (\ref{eq:P_multi}), (\ref{eq:P_multi2}). Substituting the multirate system matrices $A_k = A$, $B_k = B$, $C_k = V_k C$, and $D_k = V_k D$ into (\ref{pre0-eqcyclic}), (\ref{pre0-checkABCD}) yields the following time-invariant representation:
\begin{equation}
 \label{eqcyclic}
     \begin{array}{rcl}
          \check{x}(k+1) & = & \check{A}\check{x}(k) + \check{B}\check{u}(k)+\check{d}_u(k)\\
          \check{y}(k) & = & \check{C}\check{x}(k) + \check{D}\check{u}(k)+\check{D}_V\check{d}_y(k), 
     \end{array}
\end{equation}
where $\check{x}(k)\in \mathbb{R}^{Mn}$, $\check{y}(k)\in \mathbb{R}^{Ml}$, and
\begin{eqnarray}
\check{A} &=& \check{S}_n^{-1}(I_M \otimes A), \quad
\check{B} = \check{S}_n^{-1}(I_M \otimes B), \nonumber \\
\check{C} &=& \mathrm{diag}(V_0 C, \ldots, V_{M-1} C), \nonumber \\
\check{D} &=& \mathrm{diag}(V_0 D, \ldots, V_{M-1} D). \label{eqcyclic-mat}
\end{eqnarray}
The cycled noise signals $\check{d}_u(k) \in \mathbb{R}^{Mn}$ and $\check{d}_y(k) \in \mathbb{R}^{Ml}$ are defined consistently with the cycled input and output as $\check{d}_u(k) = \check{S}_n^{-1}(e_{k \bmod M} \otimes d_u(k))$ and $\check{d}_y(k) = e_{k \bmod M} \otimes d_y(k)$. The matrix $\check{D}_V = \mathrm{diag}(V_0, V_1, \ldots, V_{M-1}) \in \mathbb{R}^{Ml \times Ml}$ is the observation-noise feedthrough matrix, where the subscript $V$ indicates that it is built from the selection matrices $V_k$.

The time-invariant representation (\ref{eqcyclic}) describes the multirate system (\ref{eq:P_multi}), (\ref{eq:P_multi2}) exactly: with the cycled input generated by (\ref{checku}) and a matching initial condition, the cycled state and output reproduce at every $k$ the state and output of (\ref{eq:P_multi}), (\ref{eq:P_multi2}) at block position $k \bmod M$, so that nothing is lost or added by the reformulation. This holds under three conditions: (i) the observation periods $M_1,\ldots,M_l$ are known integer multiples of the control period, so that $M$ is finite and (\ref{eq:Sk2}) holds; (ii) the phase of the observation pattern is known, so that the same phase index is used in the cycling operation and in the selection matrices; and (iii) the output in (\ref{eq:P_multi2}) is the zero-padded measured output. No assumption on the plant $(A,B,C,D)$ is involved, and the property is unaffected by noise, since the cycled noise signals are defined by the same cycling operation.

For system identification, the cycled output signal $\check{y}(k)$ is constructed from the noise-affected measurements $y(k)$ by applying the same signal cycling operation as in (\ref{checku}). Hence, the effects of both process noise $d_u(k)$ and observation noise $d_y(k)$ are implicitly contained in the cycled output data used for identification.

\subsection{Controllability and Observability}
The cycled system (\ref{eqcyclic}) inherits controllability and observability from the original system under mild conditions~\cite{okajima2}.

The controllability matrix for system (\ref{eqcyclic}) is:
\begin{eqnarray}
{\Psi}_c = \left[\check B, \check A\check B, \cdots, \check A^{Mn-1}\check B\right]
\end{eqnarray}
The controllability of system (\ref{eqcyclic}) is automatically satisfied when the pair $(A,B)$ is controllable. Then, the following condition is satisfied for the cycled system:
\begin{eqnarray}
{\mathbf{rank}} {\Psi}_c = Mn
\end{eqnarray}

Next, an observability matrix for the $M$-periodic multirate system can be written as follows:
\begin{eqnarray}
{\Psi}_o= \left[\begin{array}{c}\check C\\ \check C\check A\\ \vdots\\ \check C\check A^{Mn-1}\end{array}\right]
\end{eqnarray}
The matrix size of ${\Psi}_o$ is $M^2 ln\times Mn$, since $\check C \in \mathbb{R}^{Ml \times Mn}$ and $\Psi_o$ stacks $Mn$ block rows. In multirate systems, $\check C$ contains zero blocks at non-observation timings due to $V_k$. Since $A$ has full rank $n$ by assumption, observability can be maintained despite these missing observations. Under Assumption \ref{ass2} and the full-rank assumption on $A$, the observability of the pair $(\check C,\check A)$ is established~\cite{okajima2}:
\begin{eqnarray}
{\mathbf{rank}} {\Psi}_o = Mn
\end{eqnarray}

\subsection{Auxiliary Matrices for Coordinate Transformation}\label{sec32}

This section introduces auxiliary matrices required for the coordinate transformation that converts the identified dense matrices into cyclic reformulation structure. The coordinate transformation matrix $T$, constructed using these auxiliary matrices, must be full-rank to ensure the transformation is invertible.

Define the block diagonal matrix $\check G_\nu \in \mathbb{R}^{Mm \times Mn}$ for $\nu = 0, \ldots, n-1$ as:
\begin{equation}
     \check{G}_\nu = \mathrm{diag}(\underbrace{G_\nu, \ldots, G_\nu}_{M}) \in \mathbb{R}^{Mm \times Mn}, \label{gj}
\end{equation}
where $G_\nu$ are $m \times n$ matrices. Using $\check G_\nu$, we define the candidate transformation matrix $\check Y \in \mathbb{R}^{Mn \times Mn}$ as:
\begin{eqnarray}
\check Y = \sum_{i = 0}^{n-1}\sum_{j = 0}^{M-1}  \check A^{Mi+j}\check B \check S_m^{j+1} \check G_{i}.  \label{checky2b}
\end{eqnarray}
The double sum in (\ref{checky2b}) enumerates the exponents $Mi+j$ for $i=0,\ldots,n-1$ and $j=0,\ldots,M-1$, i.e., exactly the powers $\check A^{0},\ldots,\check A^{Mn-1}$ that enter the controllability matrix of the expanded system of order $Mn$, so that $\check Y$ is a linear combination of the block columns of $\Psi_c$. Within each term the auxiliary matrix $\check G_i$ is common to the $M$ phases by (\ref{gj}), which is consistent with the rank condition (\ref{rankg}) below involving only the $n$ matrices $G_0,\ldots,G_{n-1}$ rather than $Mn$ of them.

\begin{lemma}\label{lemma4}
For any block diagonal matrices $\check G_\nu$ of the form (\ref{gj}), the product $\check C\check Y$ is a block diagonal matrix with $l \times n$ block elements.
\end{lemma}
\begin{proof}
Note that $\check H(Mi+j+1) = \check C \check A^{Mi+j}\check B$ from the definition (\ref{markov}). From Lemma \ref{lemma3}, the matrix $\check S_l^{Mi} \check H(Mi+j+1)\check S_m^{j+1}$ is block diagonal; since $\check S_l^{M} = I_{Ml}$ gives $\check S_l^{Mi} = I_{Ml}$, this implies that $\check C \check A^{Mi+j}\check B \check S_m^{j+1}$ itself has block diagonal structure. Since $\check G_{i}$ is also block diagonal, the product $\check C \check A^{Mi+j}\check B \check S_m^{j+1} \check G_{i}$ inherits the block diagonal structure. The summation over $i$ and $j$ preserves this structure.
\end{proof}

To avoid an a priori rank deficiency in the construction (\ref{checky2b}), we impose on $G_\nu$ the rank condition
\begin{eqnarray}\label{rankg}
{\mathbf{rank}} \left[\begin{array}{c}
G_0\\ \vdots \\ G_{n-1}
\end{array}\right]  = n.
\end{eqnarray}
This condition is easily met; for example, the stacked matrix $[G_0^T, \ldots, G_{n-1}^T]^T \in \mathbb{R}^{nm \times n}$ can be chosen to have rank $n$, a simple choice for single-input systems ($m = 1$) being $G_\nu$ equal to the $(\nu+1)$-th row of $I_n$ ($\nu = 0,\ldots,n-1$), so that the stacked matrix is $I_n$ itself. Whether the resulting $\check Y$ is nonsingular still depends on the specific choice of $G_\nu$ and on the controllability properties of $(\check A, \check B)$, and is checked numerically once $G_\nu$ has been fixed.

\section{Robust Model Construction via Cyclic Reformulation}\label{sec4}
\subsection{System Identification Using Cycled Signals}\label{sec41} 
This section describes how subspace identification methods~\cite{id1,id2} are applied to cycled signals for multirate system identification. Subspace methods are well-suited for this purpose because they handle multiple-input multiple-output (MIMO) systems naturally and employ numerically stable algorithms (SVD, QR decomposition).

The identification procedure is as follows:
\begin{enumerate}
\item Apply input $u(k)$ to the multirate system (\ref{eq:P_multi}), (\ref{eq:P_multi2}) and obtain the noise-affected output $y(k)$.
\item Construct cycled signals $\check u(k) \in \mathbb{R}^{Mm}$ and $\check y(k) \in \mathbb{R}^{Ml}$ from the input--output data $(u,y)$ using the cycling operation (\ref{checku}).
\item Apply subspace identification to the cycled signals $(\check u, \check y)$.
\end{enumerate}
A key observation is that applying subspace identification to the cycled signals yields a state-space model of expanded dimension $Mn$, not the original dimension $n$. Specifically, the identified parameters $(\A_*,\B_*,\C_*,\D_*)$ have dimensions $\A_* \in \mathbb{R}^{Mn\times Mn}$, $\B_* \in \mathbb{R}^{Mn\times Mm}$, $\C_* \in \mathbb{R}^{Ml\times Mn}$, and $\D_* \in \mathbb{R}^{Ml\times Mm}$. Although the cycled signals $\check u(k)$ and $\check y(k)$ constructed in Step 2 retain their block-sparse support even under noise, these expanded-dimension matrices are generally dense and do not exhibit the cyclic reformulation structure. The coordinate transformation in Section \ref{sec42} will convert these dense matrices into the structured cyclic form.

To characterize the identified parameters $(\A_*,\B_*,\C_*,\D_*)$, we examine their Markov parameters:
\begin{eqnarray}
\check \Hm^* (i) = \begin{cases}\D_*,&i = 0 \\ \C_*  \A_*^{i-1} \B_*, & i=1,2,\cdots \end{cases} 
\end{eqnarray}
In the same manner as Lemma \ref{lemma3}, we consider the following $Ml\times Mm$ matrix:
\begin{eqnarray}
\check S_l^{i} \check \Hm^* (i+j) \check S_m^j. \label{mar12}
\end{eqnarray}

Based on this characterization, we introduce the following assumption for the identified parameters $(\A_*,\B_*,\C_*,\D_*)$:

\begin{ass}\label{ass1}
(Structural preservation property) Let $\check{\Hm}^*(i)$ denote the Markov parameters of the realization $(\A_*, \B_*, \C_*, \D_*)$ obtained by subspace identification from the cycled signals $\check{u}(k), \check{y}(k)$. The structural relations established in Lemma~\ref{lemma3} for $\check{H}(i)$ are inherited by $\check{\Hm}^*(i)$: for any non-negative integers $i, j$, the matrix $\check{S}_l^i \check{\Hm}^*(i+j) \check{S}_m^j$ is block diagonal with $l \times m$ block elements, and $\check{S}_l^i \check{\Hm}^*(i+j) \check{S}_m^{j-1}$ ($j \geq 1$) is a cyclic matrix with $l \times m$ block elements. $\hfill \Box$
\end{ass}

Assumption~\ref{ass1} holds exactly in the noise-free case under the standard rank and persistence-of-excitation conditions required by the subspace identification method. If the identified realization is minimal and input--output equivalent to the true cycled system $(\check{A}, \check{B}, \check{C}, \check{D})$, then $\check{\Hm}^*(i) = \check{H}(i)$ for all $i$, so that the block-diagonal and cyclic structural relations of Lemma~\ref{lemma3} carry over verbatim. Under finite noisy data, this exact equality no longer holds and the admissible Markov blocks are perturbed by process and observation noise. However, the cycling operation imposes a phase-dependent support structure on both the input and output signals: at each time instant, only the block corresponding to $k \bmod M$ is active, and the zero-padded missing outputs remain zero in the constructed cycled signals. Consequently, the blocks of $\check{\Hm}^*(i)$ that Lemma~\ref{lemma3} requires to vanish are exactly those that would map nonzero cycled-input components to cycled-output components that are structurally zero: at any time $k$ with $k \bmod M \neq p$, the $p$-th output block of $\check y(k)$ is identically zero, including its noise component, whereas the cycled-input components that these blocks multiply are nonzero at such times. Under persistently exciting input, the cycled data therefore constrain these blocks to zero regardless of the noise level; Assumption~\ref{ass1} concerns whether the subspace algorithm returns a realization that satisfies this constraint, and its validity under noise is examined numerically later.

\subsection{Derivation of Cyclic Reformulation by Coordinate Transformation}\label{sec42}
The matrix parameters are obtained as $\A_*, \B_*, \C_*, \D_*$ by using the subspace identification method with the cycled signals $\check u(k)$ and $\check y(k)$. The identified matrices $\A_*, \B_*, \C_*, \D_*$ are generally dense and do not necessarily exhibit the cyclic reformulation structure. In this section, we introduce a state coordinate transformation for the identified realization $(\A_*, \B_*, \C_*, \D_*)$ so as to recover the cyclic reformulation structure.

For a nonsingular matrix $T \in \mathbb{R}^{Mn\times Mn}$, define the transformed realization by
\begin{eqnarray}
\check A_m&=T^{-1}\A_* T&, \check B_m=T^{-1}\B_*\label{eq:T1}\\
\check C_m&=\C_*T&, \check D_m=\D_*\label{eq:T2}
\end{eqnarray}
Then the transformed state is given by $\check x_{tf}=T^{-1}\check x_*$, and the corresponding state-space model is
\begin{equation}
     \label{cycsolution}
     \begin{array}{rcl}
          \check{x}_{tf}(k+1) & = & \check A_m\check{x}_{tf}(k) + \check B_m\check{u}(k)\\
          \check{y}(k) & = & \check C_m \check{x}_{tf}(k) + \check{D}_m\check{u}(k).
     \end{array}
\end{equation}

To convert the dense identified parameters $(\A_*, \B_*, \C_*, \D_*)$ into cyclic reformulation structure, we define the coordinate transformation matrix $T$ based on Lemma~\ref{lemma4} and Assumption~\ref{ass1} as follows:
\begin{eqnarray}
T = \sum_{i = 0}^{n-1}\sum_{j = 0}^{M-1}  \A_*^{Mi+j}\B_* \check S_m^{j+1} \check G_{i}.  \label{checky3}
\end{eqnarray}
Here, $\check G_\nu$ is given in (\ref{gj}), where $G_\nu$ satisfies the rank condition (\ref{rankg}). The expression (\ref{checky3}) is obtained from (\ref{checky2b}) by replacing $(\check A, \check B)$ with $(\A_*, \B_*)$. In the proposed algorithm, the nonsingularity of $T$ is verified numerically after selecting $G_\nu$. If $T$ is rank-deficient or ill-conditioned, $G_\nu$ is reselected.

The following theorem is a restatement, in the notation of the present paper, of the cyclic-structure recovery result proved in~\cite{okajima2}; we recall it only to make clear how the $M$ parameter sets are extracted from the identified realization. Since the statement coincides with the corresponding result of~\cite{okajima2} up to the present notation, it is presented without a separate proof, and the reader is referred to~\cite{okajima2} for the full derivation.

\begin{theorem}\label{theo1}
Let $\A_*, \B_*, \C_*, \D_*$ be parameters obtained from subspace identification of the cycled signals, and assume that Assumption~\ref{ass1} holds for these parameters. Assume further that the realization $(\A_*,\B_*,\C_*,\D_*)$ is minimal, i.e., the pairs $(\A_*, \B_*)$ and $(\C_*, \A_*)$ are controllable and observable, respectively. Suppose that $G_\nu$ satisfies the rank condition (\ref{rankg}) and that the resulting matrix $T$ in (\ref{checky3}) is nonsingular. Then the system $\check A_m, \check B_m, \check C_m, \check D_m$, obtained by the state coordinate transformation (\ref{eq:T1}), (\ref{eq:T2}) of $\A_*, \B_*, \C_*, \D_*$ using the transformation matrix $T$ of (\ref{checky3}), has the exact cyclic reformulation structure. Specifically, $\check A_m$ and $\check B_m$ are cyclic matrices, and $\check C_m$ and $\check D_m$ are block diagonal matrices.
\end{theorem}

Theorem~\ref{theo1} rests on Lemma~\ref{lemma4}, proved above, and on Lemma~\ref{lemma3}, proved in~\cite{okajima1}; both lemmas concern the true cycled system $(\check A,\check B,\check C,\check D)$, whereas the theorem is stated for the identified realization $(\A_*,\B_*,\C_*,\D_*)$, and the step that bridges the two is the following. Assumption~\ref{ass1} asserts precisely that the Markov parameters of the identified realization satisfy the structural relations of Lemma~\ref{lemma3}. Since the proof of Lemma~\ref{lemma4} uses only these relations and the block-diagonal form (\ref{gj}) of the auxiliary matrices, it applies verbatim with $(\A_*,\B_*,\C_*)$ in place of $(\check A,\check B,\check C)$, so that $\check C_m=\C_*T$ is block diagonal; the cyclic form of $\check A_m$ and $\check B_m$ then follows from the cyclic-matrix relations in Assumption~\ref{ass1} by the argument of~\cite{okajima2}. The theorem does not depend on the state coordinate returned by subspace identification: if $(\A_*,\B_*,\C_*)$ is replaced by $(P^{-1}\A_*P,\,P^{-1}\B_*,\,\C_*P)$ for a nonsingular $P$, then $T$ in (\ref{checky3}) is replaced by $P^{-1}T$ and the transformed realization (\ref{eq:T1}), (\ref{eq:T2}) is unchanged; the $M$ parameter sets are thus obtained in one common state coordinate, which is what makes their centroid and convex hull meaningful. For later reference, the assumptions required for this parameter recovery are those of Section~\ref{sec2} and Assumption~\ref{ass2} (on the plant and the data) together with minimality, Assumption~\ref{ass1}, the rank condition (\ref{rankg}), and the nonsingularity of $T$ (on the identified realization); in the noise-free case the recovery is exact up to a state coordinate transformation common to all phases, and none of these assumptions is needed for the representation equivalence of Section~\ref{sec22}.

\begin{remark}\label{rem:finiteness}
Theorem~\ref{theo1} is an exact statement based on Assumption~\ref{ass1}; it does not claim that finite noisy data yield an exactly cyclic realization, nor does it provide a guaranteed bound on the distance between the true plant and the polytope. In practice, however, the coordinate transformation (\ref{checky3}) can still be applied to noisy identified parameters $(\A_*,\B_*,\C_*,\D_*)$, provided that the structural relations in Assumption~\ref{ass1} hold to sufficient numerical accuracy. If this is the case, the transformed realization is approximately cyclic, and the $M$ parameter sets, which coincide in the noise-free case, exhibit noise-induced dispersion that can be exploited for the centroid and polytopic models in Sections~\ref{sec43a} and~\ref{sec43b}. Numerical evidence supporting Assumption~\ref{ass1} under finite noise is provided in the Appendix for the example in Section~\ref{sec52}, and the resulting models are further assessed through the output-level validation in Section~\ref{sec5}. This distinction separates the exact cyclic-recovery theorem from the practical finite-noise modeling step. The dispersion arises because each of the $M$ sets is extracted from a different phase of the cycled signal and is therefore estimated from a different, though not necessarily independent, realization of the noise; for the LTI case the same mechanism is stated in~\cite[Sec.~IV-F, Remark~1]{okajima3}, and the present remark is its multirate counterpart.
\end{remark}

Under Theorem~\ref{theo1}, the transformed matrices take the following forms:
\begin{eqnarray}
\check{A}_m &=& \check{S}_n^{-1} \mathrm{diag}(A_{m0}, A_{m1}, \ldots, A_{m(M-1)}), \label{Amtilde} \\
\check{B}_m &=& \check{S}_n^{-1} \mathrm{diag}(B_{m0}, B_{m1}, \ldots, B_{m(M-1)}), \label{Bmtilde}
\end{eqnarray}

\begin{eqnarray}
\label{Cmtilde}
\check C_{m}=\mathrm{diag}\left(C_{m0}, C_{m1}, \ldots, C_{m(M-1)}\right)
\end{eqnarray}
\begin{eqnarray}
\label{Dmtilde}
\check D_{m}=\mathrm{diag}\left(D_{m0}, D_{m1}, \ldots, D_{m(M-1)}\right)
\end{eqnarray}
The $M$ sub-block parameter sets $(A_{mi}, B_{mi}, C_{mi}, D_{mi})$, $i = 0, \ldots, M-1$, extracted from (\ref{Amtilde})--(\ref{Dmtilde}), where $A_{mi} \in \mathbb{R}^{n\times n}$, $B_{mi} \in \mathbb{R}^{n\times m}$, $C_{mi} \in \mathbb{R}^{l\times n}$, and $D_{mi} \in \mathbb{R}^{l\times m}$ match the dimensions of the original matrices in (\ref{siki1}), (\ref{siki12}), constitute the identified models that address Problem~\ref{prob1}. In the exact noise-free case these $M$ sets coincide.

Note that the $M$ parameter sets $(A_{mi}, B_{mi}, C_{mi}, D_{mi})$ are determined up to a common state coordinate transformation, consistent with the formulation in Problem~\ref{prob1}. This coordinate ambiguity is inherent to subspace identification and does not affect the centroid model or the polytopic uncertainty model constructed in the following sections.

Finally, we note how noise reaches the $M$ parameter sets. Noise enters the cycled output through the cycling operation (\ref{checku}) applied to the noisy measurements; the subspace step returns perturbed matrices $(\A_*,\B_*,\C_*,\D_*)$; and the transformation matrix $T$ in (\ref{checky3}) is itself computed from these perturbed matrices. The $M$ extracted sets therefore differ by perturbations that reflect noise acting through the whole pipeline, including through $T$. Rather than propagating these perturbations analytically, the polytope of Section~\ref{sec43b} uses them as observed.

\subsection{Centroid Model Construction}\label{sec43a}

To obtain a single noise-reduced nominal model from the $M$ parameter sets, which disperse under noise, this subsection constructs the \emph{centroid model} by averaging them.

The centroid model is defined by averaging the $M$ parameter sets extracted from $\check A_m, \check B_m, \check C_m, \check D_m$ in (\ref{Amtilde})--(\ref{Dmtilde}). Since $A_{mi}$ and $B_{mi}$ are identified at every phase $i = 0, \ldots, M-1$, they are averaged uniformly:
\begin{eqnarray}
A_c = \frac{1}{M}\sum_{i=0}^{M-1} A_{mi}, \quad  B_c = \frac{1}{M}\sum_{i=0}^{M-1} B_{mi}. \label{eq:centroid_AB}
\end{eqnarray}
For $C_{mi}$ and $D_{mi}$, the $r$-th row is identifiable only at the $M/M_r$ phases satisfying $v_r(i) = 1$, where $v_r(i)$ denotes the $r$-th diagonal entry of $V_i$ in (\ref{eq:Sk}), i.e., the value taken by $v_r(k)$ at any $k$ with $k \bmod M = i$, and $M_r$ is the observation period of output $y_r$ defined in Section~\ref{sec2}. The centroid is therefore defined row by row using only these observable-timing estimates:
\begin{eqnarray}
[C_c]_{r,:} &=& \frac{M_r}{M}\sum_{v_r(i)=1} [C_{mi}]_{r,:}, \nonumber\\
{[D_c]}_{r,:} &=& \frac{M_r}{M}\sum_{v_r(i)=1} [D_{mi}]_{r,:}, \label{eq:centroid_CD}
\end{eqnarray}
where $[\,\cdot\,]_{r,:}$ denotes the $r$-th row and the sum runs over the $M/M_r$ phases at which output $r$ is observed. This row-wise convention is justified because $C$ and $D$ are time-invariant in the original system (\ref{siki1}), (\ref{siki12}), so every observable-timing estimate of row $r$ targets the same row of $C$; averaging only the informative entries avoids double-counting that would arise if non-observation phases were filled and then re-averaged uniformly.

Because $\check C$ and $\check D$ in (\ref{eqcyclic-mat}) are built from the products $V_iC$ and $V_iD$, the $r$-th rows of $C_{mi}$ and $D_{mi}$ at non-observation phases ($v_r(i)=0$) are structurally zero. In the identified realization they vanish to machine precision, since the corresponding entries of the cycled output $\check y(k)$ are identically zero; these rows are therefore not poorly estimated but empty, which is why the row-wise completion below is required. To retain the $M$-vertex structure required for the polytopic model in Section~\ref{sec43b}, the vertex matrices $C_{mi}, D_{mi}$ at non-observation phases ($v_r(i) = 0$) are completed row by row by substituting the row centroids $[C_c]_{r,:}, [D_c]_{r,:}$ in (\ref{eq:centroid_CD}). This imputation affects only the output matrices at non-observation timings and does not influence the identified $A$ and $B$ matrices.

Let $(A_{\mathrm{ref}}, B_{\mathrm{ref}}, C_{\mathrm{ref}}, D_{\mathrm{ref}})$ denote the realization obtained from noise-free data by the procedure of Sections~\ref{sec41} and~\ref{sec42}, at which the $M$ sets coincide by Theorem~\ref{theo1}. Under noise, each set can be written as $(A_{mi}, B_{mi}, C_{mi}, D_{mi}) \approx (A_{\mathrm{ref}}, B_{\mathrm{ref}}, C_{\mathrm{ref}}, D_{\mathrm{ref}}) + \epsilon_i$, where $\epsilon_i$ is the noise-induced variation; the relation is approximate in the coordinate of the reference realization, since $T$ in (\ref{checky3}) is itself computed from noisy data. The averaging operation can then reduce noise effects if the variations $\epsilon_i$ are approximately centered around the reference realization. In the idealized case where the perturbations $\epsilon_i$ are mutually uncorrelated and have a common covariance, the variance of the centroid estimate for $A_c, B_c$ is reduced by a factor of $1/M$ compared to any single parameter set, while for each row of $C_c, D_c$ the reduction factor is $M_r/M$. Such decorrelation may hold approximately because each parameter set is extracted from a distinct cyclic sub-block processing temporally interleaved data, whereas exact statistical independence is not guaranteed since the subspace identification operates on the full cycled signal jointly; the variance-reduction factors should therefore be interpreted as idealized references rather than exact guarantees.

\subsection{Polytopic Uncertainty Model Construction}\label{sec43b}

The same $M$ parameter sets can also be used as vertices of a polytopic uncertainty model, extending the method of~\cite{okajima3} to multirate systems: while the centroid model in Section~\ref{sec43a} provides a single nominal model, the polytopic model captures the dispersion of the parameter estimates.

To define well-posed vertex matrices, the row-wise imputation in Section~\ref{sec43a} is first applied: for each output row $r$ and each non-observation phase ($v_r(i) = 0$), the $r$-th rows of $C_{mi}$ and $D_{mi}$ are replaced by the row centroids $[C_c]_{r,:}, [D_c]_{r,:}$ in (\ref{eq:centroid_CD}). With the $M$ completed parameter sets $(A_{mi}, B_{mi}, C_{mi}, D_{mi})$, $i = 0, \ldots, M-1$, fixed as the vertices, the multirate polytopic uncertainty model is obtained by instantiating the convex combinations (\ref{eq:polytope_construction}), (\ref{eq:polytope_construction_cd}) of Section~\ref{sec02} at $N = M$:
\begin{eqnarray}
A(\lambda) &=& \sum_{i=0}^{M-1} \lambda_i A_{mi}, \quad
B(\lambda) = \sum_{i=0}^{M-1} \lambda_i B_{mi}, \nonumber\\
C(\lambda) &=& \sum_{i=0}^{M-1} \lambda_i C_{mi}, \quad
D(\lambda) = \sum_{i=0}^{M-1} \lambda_i D_{mi}, \label{eq:multirate_polytope}
\end{eqnarray}
together with the state-space form (\ref{siki1-jizen2}), (\ref{siki2-jizen2}) and $\lambda \in \Lambda$ of Section~\ref{sec02} with $N = M$. Unlike in~\cite{okajima3}, the number of vertices is not a design parameter but is fixed at $M$ by the multirate structure. The polytope is a data-derived description of the dispersion of the $M$ sets; it is meaningful when Assumption~\ref{ass1} holds numerically under noise and when the vertex dispersion is produced by the phase-wise noise realizations, as stated in Remark~\ref{rem:finiteness} following~\cite[Sec.~IV-F, Remark~1]{okajima3}.

A structural consequence of the row-wise imputation should be noted. For each output row $r$, the vertices differ only through the $M/M_r$ observable-phase estimates, while the imputed rows coincide with the row centroid; the dispersion of $C(\lambda)$ and $D(\lambda)$ along rarely observed output rows is therefore structurally smaller than that of $A(\lambda)$ and $B(\lambda)$. In the extreme case $M_r = M$, where row $r$ is observed at a single phase, the corresponding row is identical at all vertices and carries no dispersion in the polytope. The polytope may thus underrepresent the uncertainty of such output rows, which should be taken into account when it is used for robust control design.

By construction, the polytope (\ref{eq:multirate_polytope}) contains the centroid model as a specific interior point: setting $\lambda_i = 1/M$ for all $i$ recovers $A(\lambda) = A_c$ and $B(\lambda) = B_c$ directly from (\ref{eq:centroid_AB}), and the row-wise imputation ensures that each row of $C(\lambda)$ and $D(\lambda)$ at $\lambda_i = 1/M$ also coincides with $[C_c]_{r,:}$ and $[D_c]_{r,:}$ in (\ref{eq:centroid_CD}). The centroid and polytopic models are therefore consistent within the same recovered realization: the centroid serves as the nominal model at the polytope center, and the surrounding vertices describe the data-derived dispersion around it.

Under a common Lyapunov matrix, LMI conditions verified at the $M$ vertices hold for all points inside the polytope by convexity (Section~\ref{sec02}), which enables direct application of $H_\infty$ control~\cite{robust1}, robust pole placement~\cite{robust3,robust4}, and robust model predictive control~\cite{kothare1996robust}, including recent formulations under polytopic model uncertainty~\cite{chen2024slsmpc}. The effectiveness of LMI-based robust $H_\infty$ design on a polytope constructed from the same cyclic reformulation principle has been numerically investigated, in the LTI setting, through Monte Carlo trials in~\cite{okajima3}: a robust $H_\infty$ state-feedback gain was designed by LMI at the polytope vertices, and over ten independent noise realizations the synthesis was feasible in every trial and the designed controller stabilized the true plant~\cite[Sec.~VI]{okajima3}. The present paper addresses the identification stage of this pipeline, and the LMI-based synthesis stage is treated in~\cite{okajima3}.

\subsection{System Identification Algorithm}\label{sec43}

Based on the above, the identification algorithm for multirate sensing systems is summarized as Algorithm \ref{algo11}.

\begin{algorithm}[!bth]
\caption{System Identification for Multirate Sensing Systems}
\label{algo11}
\begin{algorithmic}[1]
\State Determine $M$ from each output period $M_i$ of the multirate system and prepare cycled input and cycled output signals from the input--output data obtained from system (\ref{eq:P_multi}), (\ref{eq:P_multi2}).
\State Compute $\A_*$, $\B_*$, $\C_*$, $\D_*$ using an existing subspace identification method with the cycled signals.
\State Choose $G_\nu \in \mathbb{R}^{m \times n}$ for $\nu = 0, \ldots, n-1$ satisfying the rank condition (\ref{rankg}); for single-input systems ($m = 1$), a simple choice is $G_\nu$ equal to the $(\nu+1)$-th row of $I_n$. Construct $T$ from (\ref{checky3}), verify $\mathrm{rank}(T) = Mn$, and if $T$ is rank-deficient or ill-conditioned, reselect $G_\nu$.
\State Derive cyclic reformulation using the obtained $\A_*$, $\B_*$, $\C_*$, $\D_*$ with the transformation matrix $T$.
\State Extract parameters $A_{mi}, B_{mi}, C_{mi}, D_{mi}$ from the components of the cyclic reformulation $\check A_m, \check B_m, \check C_m, \check D_m$.
\State Compute the centroid model: $A_c, B_c$ by uniform averaging over $i = 0, \ldots, M-1$ as in (\ref{eq:centroid_AB}), and each row of $C_c, D_c$ by averaging only over the observable phases satisfying $v_r(i) = 1$ as in (\ref{eq:centroid_CD}). For each output row $r$, the $r$-th rows of $C_{mi}, D_{mi}$ at non-observation phases are then filled with $[C_c]_{r,:}, [D_c]_{r,:}$ so that all $M$ vertices have complete output matrices for the polytopic model.
\State Construct the polytopic uncertainty model using the $M$ parameter sets as vertices.
\end{algorithmic}
\end{algorithm}

Fig.~\ref{fig:algorithm_overview} provides an overview of Algorithm~\ref{algo11}. The input--output data are first transformed into cycled signals, from which subspace identification yields expanded-dimension parameters. A coordinate transformation then recovers the cyclic reformulation structure, producing $M$ parameter sets. These sets are used both to construct the centroid model by averaging and to define the vertices of a polytopic uncertainty model.

\begin{figure}[!t]
\centering
\includegraphics[width=0.75\textwidth]{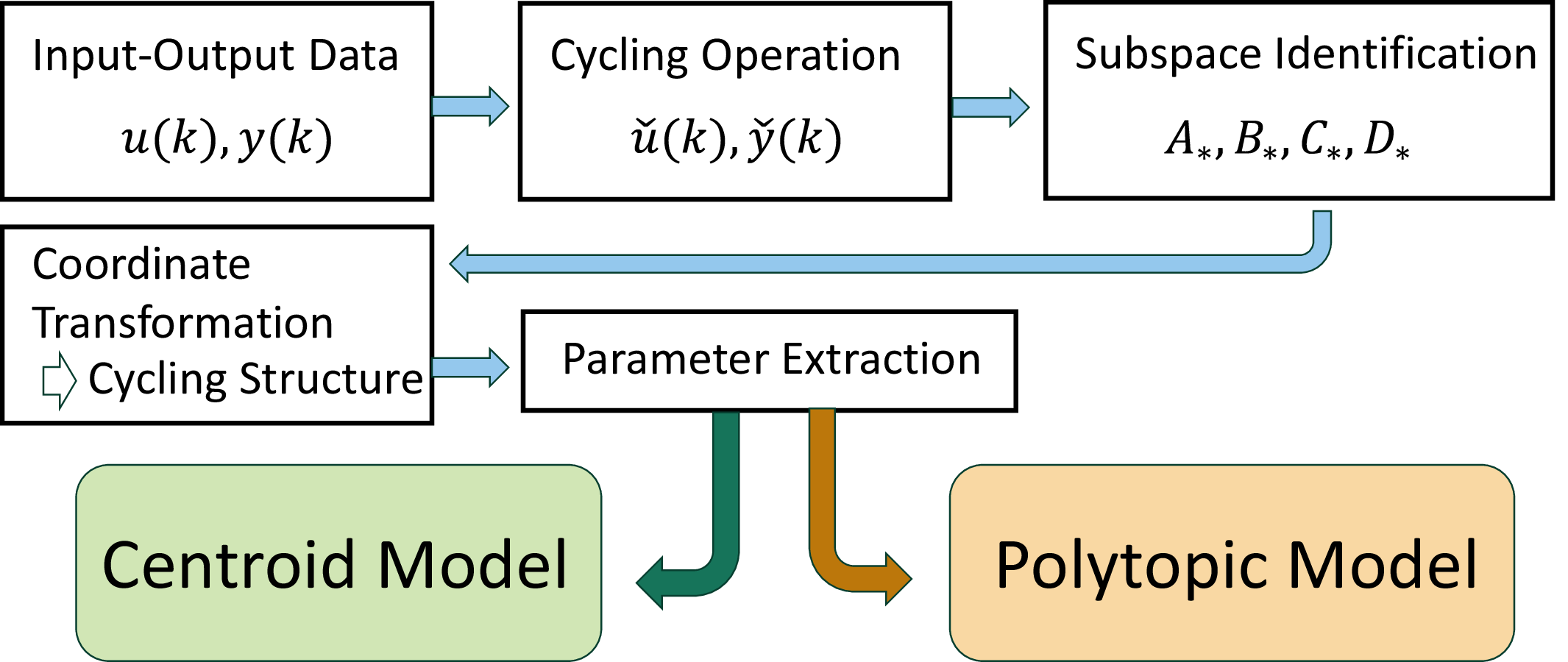}
\caption{Overview of the proposed identification algorithm. The $M$ parameter sets obtained from cyclic reformulation are used to construct both a centroid model and a polytopic uncertainty model.}
\label{fig:algorithm_overview}
\end{figure}

Steps 1--5 solve Problem~\ref{prob1} by extracting $M$ parameter sets from the cyclic reformulation, each estimating the original system matrices $(A, B, C, D)$ up to coordinate transformation. Steps 6 and 7 then produce two complementary models from these parameter sets: the centroid model serves as a noise-reduced nominal estimate of $(A, B, C, D)$, while the polytopic model enables LMI-based robust control synthesis. The computational complexity is dominated by the subspace identification in Step 2, which scales as $O(N_{\text{data}}(Mn)^2 + (Mn)^3)$ in the expanded dimension $Mn$; all other steps involve only matrix manipulations of negligible cost. The relationship to existing approaches, together with the scalability of the method with respect to $n$ and $M$, is discussed in Section~\ref{sec_disc}.

\section{Numerical Simulation}\label{sec5}

In this section, the proposed framework is evaluated through two complementary numerical examples; the finite-noise behavior is assessed via output-level validation quantities, which are invariant under similarity transformations. Section~\ref{sec51} examines the centroid nominal model introduced in Section~\ref{sec43a}: a SISO illustrative example is used to compare the centroid against the best individual vertex and an interpolation-based baseline, thereby verifying the noise-reduction effect of the row-wise averaging operation. The interpolation-based baseline is compared on both examples (Tables~\ref{tab:comparison} and~\ref{tab:interp_L6}); in both, its error is found to be a systematic modeling bias that depends only weakly on the noise level. Section~\ref{sec52} addresses the polytopic uncertainty model of Section~\ref{sec43b}: in the spirit of the LTI counterpart in~\cite{okajima3}, the best in-polytope point $\lambda^*$ is used as a diagnostic indicator of polytope reliability, namely whether the polytope contains models that approximate the true plant adequately for vertex-based LMI synthesis; the centroid's noise-reduction effect is also reconfirmed on this MIMO system. Sensitivity analyses are reported as follows: sensitivity to the noise level in Table~\ref{tab:comparison} (SISO, four levels) and Table~\ref{tab:noise_effect} (MIMO, six levels), to the data length in Table~\ref{tab:data_size_effect}, and to the sampling-period ratio in Table~\ref{table:period_effects}. Matrix-coordinate quantities such as vertex-to-vertex Frobenius norms are realization-dependent and are therefore not used as primary quantitative evidence for the polytope.

\subsection{Illustrative SISO Example and Comparison with Interpolation-Based Identification}\label{sec51}

This subsection serves two purposes: it traces Algorithm~\ref{algo11} step by step on a small example, and it compares the centroid model with an interpolation-based baseline (Table~\ref{tab:comparison}). We consider a second-order system ($n=2$) with single input ($m=1$) and single output ($l=1$). The true plant parameters used for data generation are
\begin{eqnarray}
A &=& \begin{bmatrix} 0 & 1 \\ -0.5 & 1.2 \end{bmatrix}, \quad
B = \begin{bmatrix} 0 \\ 1 \end{bmatrix}, \nonumber\\
C &=& \begin{bmatrix} 1 & 0 \end{bmatrix}, \quad
D = 0.
\label{ex1_plant}
\end{eqnarray}
The output is observed every $M=3$ steps according to
\begin{eqnarray}
V_0 = 1, \quad V_1 = 0, \quad V_2 = 0,
\label{ex1_V}
\end{eqnarray}
so that measurements are available only at $k=0,3,6,\ldots$, i.e., two-thirds of the outputs are missing. Synthetic input--output data are generated by simulating (\ref{ex1_plant}) under the multirate sampling (\ref{ex1_V}) with Gaussian process and observation noise, the input $u(k)$ being zero-mean Gaussian white noise of unit variance, generated independently in each trial. The data length is set to $N_{\mathrm{data}} = 1000$, and the noise standard deviation $\sigma$ is varied as specified below.

\subsubsection*{Application of Algorithm \ref{algo11}}

Algorithm~\ref{algo11} is applied to the generated data as follows. Cycled signals $\check u(k) \in \mathbb{R}^{3}$ and $\check y(k) \in \mathbb{R}^{3}$ are constructed from $(u,y)$ (Step~1), and the MOESP method~\cite{verhaegen1992moesp} yields the expanded-dimension parameters $\A_* \in \mathbb{R}^{6 \times 6}$, $\B_* \in \mathbb{R}^{6 \times 3}$, $\C_* \in \mathbb{R}^{3 \times 6}$, $\D_* \in \mathbb{R}^{3 \times 3}$ (Step~2). Since the system is SISO, the auxiliary matrices are chosen as $G_0 = [1\ 0]$ and $G_1 = [0\ 1]$, so that $[G_0^T, G_1^T]^T = I_2$ satisfies the rank condition (\ref{rankg}). The coordinate transformation matrix $T$ in (\ref{checky3}) is constructed and verified to have rank $Mn = 6$ (Step~3), and the three parameter sets $(A_{mi}, B_{mi}, C_{mi}, D_{mi})$ ($i = 0, 1, 2$) are extracted as in (\ref{Amtilde})--(\ref{Dmtilde}) (Steps~4 and~5). Assumption~\ref{ass2} was verified, since $(C,A^3)$ is observable; the rank condition (\ref{rankg}) and the nonsingularity of $T$ are checked numerically, while Assumption~\ref{ass1} is verified numerically in the Appendix for the MIMO example of Section~\ref{sec52}.

\subsubsection*{Noise-Free Reference Realization}

Algorithm~\ref{algo11} is first applied to noise-free data. By the noise-free recovery result of~\cite{okajima2}, the three parameter sets coincide up to numerical precision and define a single reference realization $(A_{\mathrm{ref}}, B_{\mathrm{ref}}, C_{\mathrm{ref}}, D_{\mathrm{ref}})$, whose Markov parameters match those of the original plant (\ref{ex1_plant}) so that this realization is input--output equivalent to the true plant (though, due to the coordinate transformation in Section~\ref{sec42}, generally not identical to it). The row of $C_{\mathrm{ref}}$ is identified at the only observation phase ($V_0 = 1$); the row centroid (\ref{eq:centroid_CD}) in this SISO case reduces to this single observable estimate. The reference parameters are
\begin{eqnarray}
A_{\mathrm{ref}} &=& \begin{bmatrix}
0 & -0.5000\\
1.0000 &  1.2000
\end{bmatrix}, \quad
B_{\mathrm{ref}} = \begin{bmatrix}
 1.3518\\
-1.1105
\end{bmatrix}, \nonumber\\
C_{\mathrm{ref}} &=& \begin{bmatrix}
1.7280 & 2.1036
\end{bmatrix}, \quad
D_{\mathrm{ref}} = 0,
\label{ex1_ref}
\end{eqnarray}
where the $(1,1)$ entry of $A_{\mathrm{ref}}$ is numerically zero. Although the matrix entries of (\ref{ex1_ref}) differ from those of the true plant (\ref{ex1_plant}), both realizations have the same input--output transfer function $1/(z^2 - 1.2z + 0.5)$, consistent with similarity equivalence. This reference realization serves as the comparison basis for the noisy results below.

\subsubsection*{Noisy Identification and Parameter Dispersion}

The same procedure is then applied to noisy data ($\sigma = 0.05$). Under noise, the three parameter sets no longer coincide and are interpreted as $(A_{mi}, B_{mi}, C_{mi}, D_{mi}) \approx (A_{\mathrm{ref}}, B_{\mathrm{ref}}, C_{\mathrm{ref}}, D_{\mathrm{ref}}) + \epsilon_i$ in the recovered realization, providing a concrete numerical instance of the perturbation argument in Section~\ref{sec43a}. The resulting vertices are listed in Table~\ref{tab:polytope_vertices}. Compared with the noise-free reference, noticeable variations of several percent appear in $A_{mi}(1,2)$, $A_{mi}(2,2)$, and $B_{mi}$, while the entry $A_{mi}(2,1) \approx 1$ remains nearly invariant; this invariance is a property of the specific coordinate frame obtained by the present algorithm and should not be interpreted as a coordinate-free structural property of the underlying plant. In contrast, the rows of $C_{mi}$ coincide across the three vertices by construction: the output is observed only at phase $i = 0$, and the rows at $i = 1, 2$ are filled by the imputation of Section~\ref{sec43a} with the row centroid (\ref{eq:centroid_CD}), which here reduces to the single phase-0 estimate. Table~\ref{tab:polytope_vertices} therefore lists the completed (post-imputation) vertices. The rows at $i=1,2$ before imputation vanish to machine precision, as stated in Section~\ref{sec43a}. The overall pattern of dispersion reflects how noise perturbs the degrees of freedom available in the identification within this realization.

\begin{table}[h]
\centering
\caption{Polytope vertices (after row-wise imputation) for the illustrative example ($\sigma = 0.05$, $N_{\text{data}} = 1000$)}
\label{tab:polytope_vertices}
\begin{tabular}{ccc}
\hline
Parameter & Ref.$^*$ & Vertices $(i=0,1,2)$ \\ \hline
$A_{mi}(1,1)$ & $-0.000$ & $-0.000$, $-0.000$, $0.000$ \\
$A_{mi}(1,2)$ & $-0.500$ & $-0.511$, $-0.520$, $-0.475$ \\
$A_{mi}(2,1)$ & $1.000$ & $1.000$, $1.000$, $1.000$ \\
$A_{mi}(2,2)$ & $1.200$ & $1.201$, $1.232$, $1.177$ \\
$B_{mi}(1)$ & $1.352$ & $1.360$, $1.368$, $1.326$ \\
$B_{mi}(2)$ & $-1.110$ & $-1.109$, $-1.137$, $-1.087$ \\
$C_{mi}(1)$ & $1.728$ & $1.732$, $1.732$, $1.732$ \\
$C_{mi}(2)$ & $2.104$ & $2.102$, $2.102$, $2.102$ \\
\hline
\multicolumn{3}{l}{\footnotesize $^*$Noise-free identification.}
\end{tabular}
\end{table}

\subsubsection*{Comparison with Interpolation-Based Identification}

To quantify the advantage of the cyclic reformulation, the proposed centroid model (Step~6) is compared with an interpolation-based baseline that fills missing outputs by linear interpolation and then applies standard MOESP at the original system order $n$. Table~\ref{tab:comparison} reports the FIT metric, defined as $\text{FIT} = 100 \times \left(1 - {\|y - \hat{y}\|}/{\|y - \bar{y}\|}\right)$ [\%], over 10 independent trials at four noise levels. Unless otherwise stated, $y$ in this metric denotes the complete fast-rate output generated by the true system rather than only the observed multirate samples, and $\hat{y}$ is the corresponding fast-rate prediction by the identified model. This convention is used consistently throughout Section~\ref{sec5}, including the per-output values $\mathrm{FIT}_r$ used in Section~\ref{sec52}.

\begin{table}[h]
\centering
\caption{FIT scores [\%]: proposed method vs.\ linear-interpolation-based MOESP ($n=2$, $M=3$, $N_{\text{data}}=1000$, 10 trials). ``Best vertex'' refers to the vertex model among $\{(A_{mi}, B_{mi}, C_{mi}, D_{mi})\}_{i=0}^{M-1}$ achieving the highest FIT on the validation data; this is a diagnostic quantity used for comparison only and cannot be selected in practice without access to validation data.}
\label{tab:comparison}
\begin{tabular}{cccc}
\hline
$\sigma$ & Centroid & Best vertex$^*$ & Interp.\ MOESP \\ \hline
$0.001$ & $99.97 \pm 0.01$ & $99.96 \pm 0.02$ & $68.24 \pm 2.80$ \\
$0.010$ & $99.66 \pm 0.16$ & $99.33 \pm 0.50$ & $67.93 \pm 1.23$ \\
$0.050$ & $98.36 \pm 1.09$ & $95.55 \pm 3.12$ & $67.42 \pm 2.48$ \\
$0.100$ & $96.14 \pm 1.45$ & $92.33 \pm 4.40$ & $66.42 \pm 2.53$ \\
\hline
\multicolumn{4}{l}{\footnotesize $^*$Diagnostic only; requires validation data.}
\end{tabular}
\end{table}

The proposed centroid model consistently achieves FIT scores above 96\% across all noise levels, substantially outperforming the interpolation-based approach, which remains around 66--68\% with only a weak dependence on $\sigma$. This weak dependence indicates that its accuracy is limited not by noise but by systematic modeling error: the interpolated output does not reflect the true system dynamics, introducing a bias that subspace identification cannot remove. Zero-order hold interpolation yields even lower FIT scores (approximately 58\%) and is omitted from the table. The interpolation-based baseline is used here as a representative of the fill-then-identify approaches reviewed in Section~\ref{sec1}.

Two further observations deserve emphasis. First, the centroid model outperforms the best individual vertex at every noise level, confirming the noise-reduction effect of averaging the $M$ parameter sets. Second, the performance gap between the proposed method and interpolation-based identification remains about 30 percentage points at every noise level and widens slightly as the noise level decreases, since the cyclic reformulation can exploit low-noise data while the interpolation error constitutes an irreducible floor.

\subsubsection*{Polytopic Uncertainty Model}

The same three completed parameter sets define the polytopic model (Step~7):
\begin{eqnarray}
A(\lambda) &=& \sum_{i=0}^{2} \lambda_i A_{mi}, \quad
B(\lambda) = \sum_{i=0}^{2} \lambda_i B_{mi}, \label{ex1_polytope_ab}\\
C(\lambda) &=& \sum_{i=0}^{2} \lambda_i C_{mi}, \quad
D(\lambda) = \sum_{i=0}^{2} \lambda_i D_{mi},
\label{ex1_polytope_cd}
\end{eqnarray}
where $\lambda_0 + \lambda_1 + \lambda_2 = 1$ and $\lambda_i \ge 0$.
For $i=1,2$, where $V_i=0$, the output matrices $C_{mi}$ and $D_{mi}$
are completed by the row-wise imputation described in Section~\ref{sec43a}.
In this SISO example, the output is observed only at phase $i=0$, so the
completed output matrices at the non-observation phases coincide with the
row centroid used for the nominal output matrices. Consequently, all three
vertices share the same $C$ and $D$, and the polytope expresses uncertainty
only in the $A$ and $B$ directions, the extreme case $M_r = M$ of the structural
underrepresentation discussed in Section~\ref{sec43b}.

The three vertices can be used directly in vertex-based LMI design (Section~\ref{sec43b}). Since
the centroid model corresponds to the convex combination with
$\lambda_i=1/3$ for all $i$, it lies at the geometric center of the completed
polytope, consistent with its role as the nominal model in
Section~\ref{sec43a}.

\subsection{Multirate System with Multiple Sensors}\label{sec52}

This subsection considers a more practical scenario where multiple sensors with different sampling rates are employed. A third-order system ($n = 3$) with two inputs ($m = 2$) and two outputs ($l = 2$) is used, with the following parameters:
\begin{eqnarray}
A &=& \begin{bmatrix}0&0&0.8\\1&0&0.5\\0&1&-0.4\end{bmatrix},\quad
B = \begin{bmatrix}1&0\\0&2\\0&1\end{bmatrix}, \nonumber\\
C &=& \begin{bmatrix}1&0.5&0.3\\0.1&0.3&0.7\end{bmatrix},\quad
D = \begin{bmatrix}0&0\\0&0\end{bmatrix} \label{rei01}
\end{eqnarray}
The two sensors have different observation periods: $M_1=2$ for $y_1$ and $M_2 = 3$ for $y_2$. From their least common multiple, the system period is $M = 6$. The observation timing matrices $V_i$ $(i=0,\cdots,5)$ are:
\begin{eqnarray}
V_0 &=& \mathrm{diag}(1,1),\quad V_1 = \mathrm{diag}(0,0), \label{rei02} \\
V_2 &=& \mathrm{diag}(1,0),\quad V_3 = \mathrm{diag}(0,1), \nonumber \\
V_4 &=& \mathrm{diag}(1,0),\quad V_5 = \mathrm{diag}(0,0) \label{rei03}
\end{eqnarray}
The pair $(A, B)$ is controllable, so that ${\mathbf{rank}}\, {\Psi}_c = 18$ holds. Assumption~\ref{ass2} was verified directly: the observability matrix of $(V_0C, A^6)$ has rank $n = 3$, and hence ${\mathbf{rank}}\, {\Psi}_o = 18$ also holds; the remaining assumptions of Section~\ref{sec2} hold by construction of the example.

Algorithm~\ref{algo11} is applied as follows. Input $u(k)$, generated as zero-mean Gaussian white noise of unit variance with independent channels, is applied to the multirate system (\ref{rei01})--(\ref{rei03}) to obtain output $y(k)$. Cycled signals $\check u(k) \in \mathbb{R}^{12}$ and $\check y(k) \in \mathbb{R}^{12}$ are constructed (Step~1), and the MOESP method yields the expanded-dimension parameters $\A_* \in \mathbb{R}^{18\times 18}$, $\B_* \in \mathbb{R}^{18\times 12}$, $\C_* \in \mathbb{R}^{12\times 18}$, $\D_* \in \mathbb{R}^{12\times 12}$. The auxiliary matrices $G_\nu \in \mathbb{R}^{m \times n} = \mathbb{R}^{2 \times 3}$ are chosen so that the $(\nu+1)$-th column of $G_\nu$ has all entries equal to one and the remaining columns are zero, i.e., $G_0 = [\mathbf{1}_2,\mathbf{0},\mathbf{0}]$, $G_1 = [\mathbf{0},\mathbf{1}_2,\mathbf{0}]$, $G_2 = [\mathbf{0},\mathbf{0},\mathbf{1}_2]$. The stacked matrix $[G_0^T, G_1^T, G_2^T]^T \in \mathbb{R}^{6 \times 3}$ has rank $n = 3$ and satisfies the rank condition (\ref{rankg}). The coordinate transformation matrix
\begin{eqnarray}
T = \sum_{i = 0}^{2}\sum_{j = 0}^{5}  \A_*^{6i+j}\B_* \check S_m^{j+1} \check G_{i}  \label{henkanTinv}
\end{eqnarray}
has rank~18 (Step~3), and the six parameter sets $A_{mi}, B_{mi}, C_{mi}, D_{mi}$ $(i = 0, \ldots, 5)$ are extracted from the transformed matrices (Steps~4 and~5; the centroid and the polytope follow by Steps~6 and~7).

In the noise-free case, all six parameter sets coincide up to numerical precision, confirming the exact recovery established in~\cite{okajima2}. Under noisy conditions, each parameter set is perturbed differently. As noted in Remark~\ref{rem:finiteness}, the coordinate transformation recovers the cyclic structure provided that Assumption~\ref{ass1} holds to sufficient numerical accuracy. For the present example, the Appendix shows that the residual $\rho_{\mathrm{BD}}$ remains below $2 \times 10^{-15}$ across all tested noise levels ($\sigma$ up to $0.1$), indicating that Assumption~\ref{ass1} is supported at machine precision over the entire index range relevant to the construction of $T$. On this basis, the polytopic uncertainty model is constructed from the six parameter sets as described in Section~\ref{sec43b}.

\subsubsection*{Polytope Reliability via the Best In-Polytope Point}

Following the LTI counterpart in~\cite{okajima3}, the reliability of the constructed polytope is assessed through the best in-polytope point
\begin{eqnarray}
\lambda^* = \arg\max_{\lambda \in \Lambda} \min_{1 \le r \le l} \mathrm{FIT}_r(\lambda), \label{eq:bestinpoly}
\end{eqnarray}
where $\mathrm{FIT}_r(\lambda)$ is the validation FIT of the $r$-th output channel of the in-polytope model $(A(\lambda), B(\lambda), C(\lambda), D(\lambda))$ driven by an independent noise-free validation input of length $N_{\mathrm{val}} = 1000$ drawn from the same distribution (the FIT metric is defined in Section~\ref{sec51}), and the minimum over output channels imposes a worst-output (most demanding) criterion. Equation~(\ref{eq:bestinpoly}) is approximated numerically by particle swarm optimization (PSO)~\cite{kennedy1995pso} with population size~30, $100$ iterations, inertia weight linearly decreased from $0.9$ to $0.4$, cognitive and social coefficients $c_1 = c_2 = 2.0$, and simplex constraint enforcement by normalization. We hereafter denote the validation FIT of $(A(\lambda^*), B(\lambda^*), C(\lambda^*), D(\lambda^*))$ as the \emph{best in-polytope FIT}. As discussed in~\cite{okajima3}, $\lambda^*$ requires knowledge of $y_{\mathrm{val}}$ and is therefore used strictly as an offline diagnostic indicator of polytope reliability, not as a deployable estimator. A high best in-polytope FIT indicates that the polytope contains in-polytope models with strong predictive performance, supporting its use as the uncertainty description in vertex-based LMI robust control synthesis.

\subsubsection*{Effect of Noise Level}

To evaluate robustness, 10 independent trials were conducted at 6 noise levels $\sigma \in \{0.001, 0.005, 0.01, 0.02, 0.05, 0.1\}$ with $N_{\text{data}} = 3000$. The same noise level was applied to both process and observation noise, following zero-mean Gaussian distributions. Table~\ref{tab:noise_effect} reports the best in-polytope FIT for each noise level. At low noise levels ($\sigma \leq 0.01$), the best in-polytope FIT exceeds 99.6\%, with near-perfect accuracy above 99.95\% at $\sigma = 0.001$. Even at moderate noise ($\sigma = 0.02$), the values remain above 99.2\%. At high noise levels ($\sigma = 0.1$), the FIT of 97.4\% for $y_1$ and 95.5\% for $y_2$ indicates that the polytope still contains models close to the true plant. On the same MIMO system, the centroid model of Section~\ref{sec43a} is compared with the best single vertex, i.e., the parameter set among the $M$ vertices attaining the largest worst-output validation FIT (a diagnostic that cannot be selected without validation data). Mirroring the SISO result of Section~\ref{sec51}, the centroid outperforms the best vertex on both outputs at every noise level tested (values not tabulated), with the margin widening as $\sigma$ increases; at $\sigma = 0.1$, for example, the centroid attains $99.3\%$ and $94.0\%$ for $y_1$ and $y_2$, against $92.0\%$ and $91.7\%$ for the best vertex. This confirms that the variance-reduction effect of averaging the $M$ parameter sets, established in the SISO example, carries over to the multirate multi-sensor setting. Because the best in-polytope point optimizes the worst output channel in~(\ref{eq:bestinpoly}), it need not maximize any individual output; on the more frequently observed $y_1$ the fixed centroid can attain a higher FIT than the best in-polytope model, whereas on the bottleneck output $y_2$ the ordering is reversed.

\begin{table}[h]
\centering
\caption{Effect of noise level on the best in-polytope FIT ($M=6$, $N_{\text{data}}=3000$, 10 trials)}
\label{tab:noise_effect}
\begin{tabular}{ccc}
\hline
Noise Level $\sigma$ & FIT $y_1$ [\%] & FIT $y_2$ [\%] \\ \hline
0.001 & 99.98 $\pm$ 0.01 & 99.96 $\pm$ 0.02 \\ 
0.005 & 99.87 $\pm$ 0.06 & 99.75 $\pm$ 0.09 \\ 
0.010 & 99.80 $\pm$ 0.10 & 99.69 $\pm$ 0.17 \\ 
0.020 & 99.52 $\pm$ 0.25 & 99.27 $\pm$ 0.50 \\ 
0.050 & 98.58 $\pm$ 0.95 & 98.07 $\pm$ 1.09 \\ 
0.100 & 97.37 $\pm$ 1.53 & 95.50 $\pm$ 3.48 \\ 
\hline
\end{tabular}
\end{table}

Table~\ref{tab:interp_L6} reports the interpolation-based baseline of Section~\ref{sec51} applied output-wise to the MIMO example: each output row is linearly interpolated over its own observation times before standard MOESP at order $n$, and the resulting model is evaluated on the same validation input as above. The baseline attains only about 75--78\% for $y_1$ and 63--66\% for $y_2$ with only a weak dependence on $\sigma$, whereas on the same system the best in-polytope FIT of Table~\ref{tab:noise_effect} exceeds 95\% for both outputs at every noise level and the centroid model exceeds 93\% on average at every noise level (99.3\% and 94.0\% at $\sigma = 0.1$, as stated above). As in Table~\ref{tab:comparison}, the weak dependence on $\sigma$ indicates a systematic modeling bias introduced by the interpolated samples, and its magnitude follows the observation rate: $y_2$, observed at one third of the steps, is at the level of the SISO example, whereas $y_1$, observed at every second step, is less affected. Zero-order-hold interpolation is again inferior to linear interpolation at every noise level and is therefore not tabulated.

\begin{table}[h]
\centering
\caption{Linear-interpolation-based MOESP on the MIMO example ($M=6$, $N_{\text{data}}=3000$, 10 trials, obtained from a set of trials independent of those in Table~\ref{tab:noise_effect})}
\label{tab:interp_L6}
\begin{tabular}{ccc}
\hline
$\sigma$ & FIT $y_1$ [\%] & FIT $y_2$ [\%] \\ \hline
$0.001$ & $77.29 \pm 2.33$ & $65.99 \pm 3.57$ \\
$0.005$ & $77.60 \pm 1.79$ & $66.27 \pm 2.82$ \\
$0.010$ & $77.05 \pm 1.99$ & $64.46 \pm 4.44$ \\
$0.020$ & $77.30 \pm 2.44$ & $64.06 \pm 4.34$ \\
$0.050$ & $76.58 \pm 2.96$ & $64.95 \pm 4.06$ \\
$0.100$ & $74.80 \pm 3.68$ & $62.65 \pm 5.17$ \\
\hline
\end{tabular}
\end{table}

\subsubsection*{Effect of Data Length}

Table~\ref{tab:data_size_effect} reports the best in-polytope FIT for $N_{\text{data}} \in \{500, 1000, 2000, 3000, 5000, 10000\}$ at fixed $\sigma = 0.01$. Even with limited data ($N_{\text{data}} = 500$), the polytope reliability remains high (99.1\% for $y_1$ and 98.6\% for $y_2$). Performance improves with $N_{\text{data}}$ and the gain becomes marginal beyond $N_{\text{data}} \approx 2000$, suggesting that around 300--500 samples per phase are sufficient in this example.

\begin{table}[h]
\centering
\caption{Effect of data length on the best in-polytope FIT ($M=6$, $\sigma=0.01$, 10 trials)}
\label{tab:data_size_effect}
\begin{tabular}{ccc}
\hline
Data Size $N_{\text{data}}$ & FIT $y_1$ [\%] & FIT $y_2$ [\%] \\ \hline
500 & 99.14 $\pm$ 0.41 & 98.64 $\pm$ 0.75 \\ 
1000 & 99.42 $\pm$ 0.33 & 99.17 $\pm$ 0.47 \\ 
2000 & 99.69 $\pm$ 0.12 & 99.62 $\pm$ 0.13 \\ 
3000 & 99.76 $\pm$ 0.13 & 99.64 $\pm$ 0.25 \\ 
5000 & 99.78 $\pm$ 0.12 & 99.71 $\pm$ 0.13 \\ 
10000 & 99.85 $\pm$ 0.06 & 99.71 $\pm$ 0.13 \\ 
\hline
\end{tabular}
\end{table}

\subsubsection*{Performance Across Different Sensor Configurations}

To investigate the applicability of the proposed method to various multirate sensing scenarios, we compare performance across four sensor configurations with different observation period pairs $(M_1, M_2)$. Table~\ref{table:period_effects} shows the best in-polytope FIT for these configurations, where $M = \mathrm{lcm}(M_1, M_2)$ ranges from 4 to 12. The configurations include cases where one period divides the other ($M_1 = 2, M_2 = 4$ and $M_1 = 2, M_2 = 8$) as well as coprime periods ($M_1 = 2, M_2 = 3$ and $M_1 = 3, M_2 = 4$), covering a representative range of practical multirate sensing scenarios. The $(M_1, M_2) = (2, 3)$ row was obtained from a set of trials independent of those in Table~\ref{tab:noise_effect}, so the statistics differ slightly between the two tables.

\begin{table}[h]
\centering
\caption{Best in-polytope FIT across different sensor configurations ($\sigma=0.01$, $N_{\text{data}}=3000$, 10 trials)}
\label{table:period_effects}
\begin{tabular}{cccccc}
\hline
$M_1$ & $M_2$ & $M$ & State Order & FIT $y_1$ [\%] & FIT $y_2$ [\%] \\ \hline
$2$ & $4$ & $4$ & $12$ & $99.80 \pm 0.17$ & $99.64 \pm 0.14$ \\ 
$2$ & $3$ & $6$ & $18$ & $99.92 \pm 0.03$ & $99.42 \pm 0.29$ \\ 
$2$ & $8$ & $8$ & $24$ & $99.90 \pm 0.04$ & $99.65 \pm 0.11$ \\ 
$3$ & $4$ & $12$ & $36$ & $98.96 \pm 0.63$ & $99.78 \pm 0.08$ \\ 
\hline
\end{tabular}
\end{table}

Best in-polytope FIT values above 98.9\% are achieved across all configurations, indicating that the constructed polytope contains models close to the true plant across diverse sensor period combinations. For $M = 12$ (state order~36), the FIT of $y_1$ shows slightly larger variance, which is attributable to this configuration having the largest $M$ among those tested, as discussed below. Nevertheless, the polytope reliability indicator remains high even for this demanding configuration. The sampling-period ratio affects the results through $M$. When one period divides the other, $M$ equals the larger period, whereas for coprime periods $M = M_1 M_2$; this is why the $(3,4)$ configuration has the largest state order ($36$) although its larger period is smaller than that of the $(2,8)$ configuration. A larger $M$ increases the expanded dimension $Mn$ and, for fixed $N_{\text{data}}$, reduces the data available per phase ($N_{\text{data}}/M$); among the four configurations in Table~\ref{table:period_effects}, a degradation is noticeable only for $(3,4)$, which has the largest $M$. Four configurations are too few to be conclusive, but together with Table~\ref{tab:data_size_effect} the results indicate that what matters is the data per phase rather than divisibility or coprimality as such. The resulting design implication, choosing periods that divide one another where the sensor system permits, is discussed in Section~\ref{sec_disc}.
\section{Discussion}\label{sec_disc}

The proposed method is positioned against the classes of approaches reviewed in Section~\ref{sec1} as follows. Subspace-based and stochastic-gradient methods for non-uniformly sampled or multirate data~\cite{li2006subspace,ding2014modified,ding2010partially}, lifting-based methods~\cite{rev1-1,li2001fast,hu2024beyond}, and EM-based methods~\cite{xie2013fir,xiong2014multiple,chen2019state} deliver a single nominal model, whereas the proposed method delivers, from the same single non-iterative identification run, a polytopic uncertainty model in addition to the nominal centroid model. Since the polytopic model has no counterpart among these methods, it is assessed through the best in-polytope FIT and the sensitivity analyses of Section~\ref{sec52} rather than by comparison; for the centroid model, a quantitative comparison is given against the interpolation-based baseline in Section~\ref{sec5}, and its relation to the other classes is as follows. Lifting-based methods identify a fast-rate model in lifted coordinates in which the plant matrices appear only through products such as $A^{M}$ and $A^{j}B$, so that recovering $(A,B,C,D)$ requires a matrix-root or factorization step; in the cyclic reformulation the plant matrices remain sub-blocks of the expanded matrices and are read off directly after a coordinate transformation, in the original dimension $n$ up to a common similarity transformation. EM- and maximum-likelihood-based methods~\cite{dempster1977,xie2013fir,xiong2014multiple,chen2019state,gibson2005robust} are iterative and require initialization and convergence monitoring, whereas the proposed method consists of one subspace identification on the cycled signals followed by a fixed coordinate transformation. Periodic-system identification methods~\cite{lpv,lptv,vernerey2023harmonic} identify a genuinely time-varying plant and, in the case of~\cite{lptv}, require periodic inputs; in the present setting the plant is time-invariant, the periodicity enters only through the output selection matrix $V_k$, all $M$ parameter sets estimate the same plant, and no specific periodic input is required~\cite{okajima2}. Finally, polytopic LMI-based robust control~\cite{robust1,robust3,robust4,kothare1996robust} and set-membership identification~\cite{milanese1991optimal} define the target representation rather than competing methods: the proposed polytope is a data-derived description obtained under a stochastic noise model, and it is directly usable in the vertex-based LMI framework.

Regarding computational complexity and scalability, the identification stage is a subspace identification of an LTI system of order $Mn$; for fixed $M$ it therefore scales with $n$ as standard subspace identification does. The principal practical limitation is the growth of $M$: when the observation periods are pairwise coprime, $M$ grows multiplicatively (e.g., $M=60$ for three sensors with periods $3$, $4$, and $5$), and since Table~\ref{tab:data_size_effect} suggests that the data requirement scales with the number of samples per phase (about 300--500 in that example), both the required data length and the expanded dimension increase accordingly. This can be mitigated at the sensor-system design stage by choosing periods that divide one another, for which $M$ equals the largest period. On the other hand, the polytope vertices have the original dimension $n$, so the downstream vertex-based LMI problem consists of $M$ constraints whose size is determined by $n$; the expanded dimension affects the identification stage only.

The polytopic model should be interpreted as a data-derived uncertainty description rather than a guaranteed set-membership bound~\cite{milanese1991optimal}: the vertices are constructed from noise-induced parameter variations under a fixed common state coordinate, and formal inclusion guarantees would require additional assumptions on noise bounds or probabilistic characterizations. Accordingly, a high best in-polytope FIT (Section~\ref{sec52}) demonstrates that a high-quality model exists inside the polytope, but does not by itself guarantee that the true plant lies inside it; containment would require additional assumptions on the noise. In addition, the row-wise imputation makes the dispersion of $C(\lambda)$ and $D(\lambda)$ along rarely observed output rows structurally smaller than that of $A(\lambda)$ and $B(\lambda)$, as discussed in Section~\ref{sec43b}. Under this common-coordinate construction, the polytope is directly usable as an input to vertex-based LMI robust control design~\cite{robust1,robust3,robust4,kothare1996robust}, where verifying conditions at the vertices suffices for all interior points due to convexity. The theoretical development assumes that the underlying plant is linear and time-invariant, with the time-varying structure arising solely from multirate sampling; extension to systems with intrinsic parameter variation, and experimental verification on physical multirate sensor systems, remain open problems.

\section{Conclusion}\label{sec6}
This paper has extended the cyclic reformulation-based multirate identification of~\cite{okajima2}, which achieves exact parameter recovery under noise-free conditions, to a non-iterative, control-oriented modeling pipeline that remains useful under noise by exploiting the noise-induced dispersion among the $M$ parameter sets. The $M$ parameter sets, all expressed in a single common state coordinate obtained from one identification run, are used in two complementary ways: their centroid serves as a noise-reduced nominal model, and their convex hull provides a polytopic uncertainty model directly compatible with vertex-based LMI robust control design. Numerical simulations support both models: the illustrative SISO example confirms that the centroid model attains higher validation FIT than the best individual vertex and substantially outperforms an interpolation-based baseline, while the MIMO multirate sensing example confirms, in the spirit of the LTI counterpart~\cite{okajima3}, that the best in-polytope FIT stays around 99\% at $\sigma = 0.01$ across sensor configurations with $M \in \{4,6,8,12\}$, and remains above 95\% on average even at the highest tested noise level ($\sigma = 0.1$, evaluated for the $M=6$ configuration). A practical limitation is that the least common multiple $M$, and hence the expanded dimension $Mn$ and the required data length, grow multiplicatively when many observation periods are pairwise coprime.

Future work includes a formal proof of the sparse Markov-parameter structure in Assumption~\ref{ass1} and of its stability under noise, an open problem also identified in~\cite{okajima2}, together with realization-invariant quantitative diagnostics for empirical vertex dispersion and a quantitative comparison of the centroid model with EM- and lifting-based identification. Experimental verification on physical multirate sensor systems would need to address the validity of the LTI plant assumption, non-Gaussian and correlated measurement noise, and the effect of sampling-timing jitter on the periodic pattern $V_k$ assumed in (\ref{eq:Sk2}); extension to nonlinear and intrinsically time-varying plants is a further direction.

\section*{Declaration of competing interest}
The authors declare that they have no known competing financial interests or personal relationships that could have appeared to influence the work reported in this paper.

\section*{Funding}
This work was supported by JSPS KAKENHI Grant Number 26K07555.

\section*{Author contributions}
\textbf{Hiroshi Okajima}: Conceptualization, Methodology, Formal analysis, Software, Validation, Writing -- original draft, Writing -- review \& editing, Supervision.
\textbf{Kakeru Ono}: Software, Validation.

\section*{Data availability}
The data that support the findings of this study were generated by numerical simulation. The simulation parameters are fully described in the article.

\section*{Declaration of Generative AI and AI-assisted technologies in the writing process}
The authors acknowledge the use of Claude (Anthropic) for improving the readability and language clarity of the manuscript, and for technical verification of mathematical expressions. Specifically, Claude was used to refine the English expressions throughout all sections of the manuscript, and to check dimensional consistency and numerical accuracy of equations. All technical content, research methodology, results, and conclusions are entirely the authors' own work.


\appendix
\section{Numerical Verification of Assumption \ref{ass1}}\label{app:verification}

This appendix provides numerical evidence for Assumption~\ref{ass1} using the same system and parameters as Section~\ref{sec52}. To quantify how closely the identified Markov parameters $\check{\Hm}^*(j)$ satisfy the block-diagonal structure of Lemma~\ref{lemma3}, we introduce the normalized off-diagonal-to-diagonal residual
\begin{equation}
\rho_{\text{BD}}(j) := \frac{\|\text{off-diag}(\check{\Hm}^*(j) \check{S}^j_m)\|_F}
                            {\|\text{diag}(\check{\Hm}^*(j) \check{S}^j_m)\|_F},
\end{equation}
where $\text{off-diag}(\cdot)$ and $\text{diag}(\cdot)$ extract the off-diagonal and diagonal block elements, respectively. Under Assumption~\ref{ass1}, $\rho_{\text{BD}}(j) = 0$ holds exactly. Although Assumption~\ref{ass1} involves general index pairs $(i,j)$, the case $i = 0$ suffices for the verification. Indeed, writing $\check{S}_l^i \check{\Hm}^*(i+j) \check{S}_m^j = \check{S}_l^i \bigl(\check{\Hm}^*(i+j) \check{S}_m^{i+j}\bigr) \check{S}_m^{-i}$, the map $E \mapsto \check{S}_l^i E \check{S}_m^{-i}$ merely permutes the block positions cyclically (cf.\ the shift property of $\check{S}_q$ in Section~\ref{sec01}), so that the diagonal-block and off-diagonal-block Frobenius norms are invariant and the residual for any pair $(i,j)$ coincides exactly with $\rho_{\text{BD}}(i+j)$. Moreover, the cyclic-structure component of Assumption~\ref{ass1} requires no separate verification, since $\check{S}_l^i \check{\Hm}^*(i+j) \check{S}_m^{j-1}$ is the product of the above block-diagonal matrix and $\check{S}_m^{-1}$ and is therefore cyclic whenever the block-diagonal property holds. We thus compute $\rho_{\text{BD}}(j)$ for $j = 1, \ldots, Mn = 18$, covering all index pairs with $i + j \leq Mn$, including those directly involved in the construction of the coordinate transformation $T$ in (\ref{checky3}). The verification procedure is as follows: generate input--output data from the multirate system (\ref{rei01})--(\ref{rei03}) with $N_{\text{data}} = 3000$, construct cycled signals, apply MOESP identification, and compute $\rho_{\text{BD}}(j)$.

Table~\ref{tab:verification} shows $\rho_{\text{BD}}(j)$ for $j = 1, \ldots, M$ together with the maximum over $j = 1, \ldots, Mn$ at three noise levels. All values are below $2 \times 10^{-15}$, at the level of machine precision, and are essentially independent of the noise level $\sigma$. This indicates that the blocks that Lemma~\ref{lemma3} requires to vanish are returned by MOESP as zero at machine precision for every tested noise level, rather than as a finite-sample artifact; a formal proof of this property of the algorithm remains open (Section~\ref{sec6}). By the index-reduction argument above, these results cover all index pairs $(i,j)$ with $i+j \leq Mn$, together with the cyclic-structure component implied by the block-diagonal one; only the pairs with $i+j > Mn$, which do not enter the construction of $T$, lie outside the table. Accordingly, the present results constitute numerical evidence supporting Assumption~\ref{ass1} over the entire index range required by the proposed algorithm.

\begin{table}[h]
\centering
\caption{Verification of block diagonal property: $\rho_{\text{BD}}(j)$ for $\check{\Hm}^*(j)\check{S}_m^j$ ($N_{\text{data}}=3000$; verified for all $j=1,\ldots,Mn=18$)}
\label{tab:verification}
\begin{tabular}{cccccccc}
\hline
$\sigma \backslash j$ & $1$ & $2$ & $3$ & $4$ & $5$ & $6$ & $\max_{1 \leq j \leq 18}$ \\ \hline
$0$ & $8.03\text{e-}16$ & $1.09\text{e-}15$ & $1.15\text{e-}15$ & $9.11\text{e-}16$ & $1.24\text{e-}15$ & $8.96\text{e-}16$ & $1.34\text{e-}15$ \\
$0.01$ & $5.65\text{e-}16$ & $1.23\text{e-}15$ & $1.21\text{e-}15$ & $7.81\text{e-}16$ & $1.36\text{e-}15$ & $9.44\text{e-}16$ & $1.65\text{e-}15$ \\
$0.10$ & $5.61\text{e-}16$ & $1.07\text{e-}15$ & $9.29\text{e-}16$ & $8.01\text{e-}16$ & $1.12\text{e-}15$ & $7.02\text{e-}16$ & $1.25\text{e-}15$ \\
\hline
\end{tabular}
\end{table}
\end{document}